\documentclass[pdflatex,sn-mathphys-num]{sn-jnl}

\usepackage{graphicx}%
\usepackage{multirow}%
\usepackage{amsmath,amssymb,amsfonts}%
\usepackage{amsthm}%
\usepackage{mathrsfs}%
\usepackage[title]{appendix}%
\usepackage{xcolor}%
\usepackage{textcomp}%
\usepackage{manyfoot}%
\usepackage{booktabs}%
\usepackage{algorithm}%
\usepackage{algorithmicx}%
\usepackage{algpseudocode}%
\usepackage{listings}%
\usepackage{enumitem}%

\theoremstyle{thmstyleone}%
\newtheorem{theorem}{Theorem}[section]
\newtheorem{proposition}[theorem]{Proposition}%

\newtheorem{lemma}[theorem]{Lemma}

\newtheorem{assumption}{Assumption}

\theoremstyle{thmstyletwo}%

\theoremstyle{thmstylethree}%

\newtheorem*{remark}{Remark}%

\def\OO{\boldsymbol{\Omega}}

\begin{document}

\title{
Scalable and Calibrated Sampling for Bayesian Generalized Linear Mixed Model via Stochastic Gradient Markov Chain Monte Carlo
}

\author[1]{\fnm{Youngsoo} \sur{Baek}}\email{youngsoo.baek@duke.edu}

\author[2]{\fnm{Andrea} \sur{Agazzi}}\email{andrea.agazzi@unibe.ch}

\author[3]{\fnm{Felipe} \sur{A. Medeiros}}\email{fmedeiros@med.miami.edu}

\author*[1,4]{\fnm{Samuel} \sur{Berchuck}}\email{sib2@duke.edu}

\affil*[1]{\orgdiv{Department of Biostatistics \& Bioinformatics}, \orgname{Duke University}, \orgaddress{\street{2424 Erwin Rd}, \city{Durham}, \postcode{27705}, \state{NC}, \country{USA}}}

\affil[2]{\orgdiv{Institute for Mathematical Statistics and Actuarial Sciences}, \orgname{Universit{\"a}t Bern}, \orgaddress{\street{Alpeneggstrasse 22}, \city{Bern}, \postcode{3012}, \country{Switzerland}}}

\affil[3]{\orgdiv{Department of Ophthalmology}, \orgname{University of Miami}, \orgaddress{\street{900 NW 17th St}, \city{Miami}, \postcode{33136}, \state{FL}, \country{USA}}}

\affil*[4]{\orgdiv{Department of Statistical Science}, \orgname{Duke University}, \orgaddress{\street{214 Old Chemistry}, \city{Durham}, \postcode{27708}, \state{NC}, \country{USA}}}

\abstract{
Generalized linear mixed models (GLMMs) are widely used for analyzing correlated data, particularly in large-scale biomedical and social science applications. 
Scalable Bayesian inference for GLMMs is challenging due to an intractable marginal likelihood and a high computational cost incurred by conventional Markov chain Monte Carlo (MCMC) methods. 
We develop a stochastic gradient MCMC (SGMCMC) algorithm tailored to GLMMs that enables accurate posterior inference in the large-sample regime. 
Our approach uses Fisher’s identity to construct a (biased) Monte Carlo estimator of the gradient of the marginal log-likelihood, making SGMCMC feasible when direct gradient computation is impossible. 
We analyze the additional variability, introduced by both data subsampling and gradient approximation, to derive a post-hoc covariance correction that yields properly calibrated posterior uncertainty. 
We show through simulated studies that the proposed method provides accurate posterior means and variances in settings with a large number of groups, outperforming existing approaches, including control variate methods.
We further demonstrate the method’s practical utility in an analysis of electronic health records data, where accounting for variance inflation materially changes scientific conclusions.
}



\maketitle

\section{Introduction}
\label{sec:intro}

Generalized linear mixed models (GLMMs) are widely used to model dependent data with clustered or repeated measures, particularly in biomedical and social sciences where such structures are common \citep{agresti2012categorical, breslow1993approximate}. By incorporating subject-specific random effects, GLMMs account for within-group correlation and support inference on both population- and subject-level parameters.\footnote{In this work, we use the terms ``group'' and ``subject'' interchangeably to indicate a natural unit of nesting in the dataset. The latter term is more frequently used in longitudinal contexts, whereas the former term is generally used in diverse applications, including multi-cohort/center studies.} 
Bayesian inference is attractive for GLMMs as it offers coherent uncertainty quantification and flexible hierarchical modeling \citep{gelman1995bayesian}. However, scalable Bayesian inference is challenging due to the intractable marginal likelihood, which requires integrating over random effects. One could bypass this challenge with standard Markov chain Monte Carlo (MCMC) techniques that augments the parameter space with all group-specific random effects, but this approach becomes computationally infeasible when the number of groups grows to be large. 
Marginal likelihood-based approaches, including variational inference and stochastic gradient MCMC (SGMCMC), are scalable, but require evaluating or differentiating the marginal likelihood, which is intractable in GLMMs. 
Thus, GLMMs present a unique computational bottleneck: conditional methods are tractable but not scalable, and marginal methods are scalable but intractable. 
In this paper, we demonstrate a practical approach to overcoming this difficulty, by combining two distinct ideas.
The first idea is to construct a stochastic gradient estimator using Fisher's identity and an additional layer of MCMC for the random effects.
The second idea is to correct the posterior samples to yield an asymptotically calibrated posterior variance, based on previous literature analyzing the asymptotic bias/variance of SGMCMC \cite{teh2016consistency,vollmer2016exploration}.
~

~

\noindent {\bf Variance Correction of SGMCMC.} Many SGMCMC algorithms are based on discretizing continuous-time stochastic differential equations (SDEs) whose stationary distributions coincide with the target Bayesian posteriors \citep{nemeth2021stochastic}. 
The theory of Langevin diffusion forms the foundation for understanding sampling algorithms that approximate Bayesian inference through stochastic dynamics. A typical case is provided by the SGLD algorithm of \citet{welling2011bayesian}, which replaces full-data gradients in Langevin dynamics with stochastic gradients computed from minibatches. While asymptotic consistency has been established for vanishing step sizes \citep{teh2016consistency}, constant step sizes are often used in practice for better mixing and computational efficiency \citep{li2016preconditioned}. 
However, this inflates the posterior variance relative to the true posterior \citep{brosse2018promises}, highlighting the variance inflation issue present in many SGMCMC algorithms.

A number of variance correction approaches have been proposed, including preconditioning \citep{vollmer2016exploration, stephan2017stochastic}, stochastic gradient Hamiltonian Monte Carlo \citep{chen2014stochastic}, Riemannian manifold extensions \citep{patterson2013stochastic,ma2015complete} and control variate methods \citep{baker2019control}. 
Of these, only the control variate method remains computationally comparable to the uncorrected SGLD, but so far it has not been adapted to GLMMs.
We propose a post-processing approach, complementing this suite of ideas.
It essentially adapts an extended Kalman filter-like linearization technique frequently used in applied SDE modeling literature \cite{sarkka2019applied}.
We demonstrate through simulated studies that our method yields an accurate approximation of the posterior variance, while a naive use of control variates method in GLMMs can lead to both inflated and deflated variance estimates.
~

~

\noindent {\bf Intractable Likelihoods and Gradients.} 
There are numerous methods to deal with intractable integrals appearing in GLMMs, including Laplace approximation \citep{tierney1989fully}, adaptive Gaussian-Hermite quadrature \citep{naylor1982applications} and penalized quasi-likelihood \citep{breslow1993approximate}. Traditional methods often struggle in high-dimensional or non-Gaussian random effects settings. 
Monte Carlo techniques \citep{chib1995marginal, fruhwirth2004estimating} and EM-based variants \citep{booth1999maximizing} offer more flexibility, but can be computationally intensive and sensitive to tuning parameters. 
Bayesian methods, such as MCMC \citep{gelman1995bayesian}, INLA \citep{rue2017bayesian} and variational inference \citep{ormerod2012gaussian, tan2013variational}, are more scalable but still face challenges with complex random effect structures. 

To avoid directly evaluating the marginal likelihood, we leverage Fisher’s identity \citep{fisher1925theory}.
Previously used for variational inference in GLMMs \citep{tran2020bayesian}, Fisher's identity equates the gradient of the marginal likelihood to a certain expectation of random effects with respect to its conditional distribution given the data and population parameters. 
Our method is very closely related to the recent work that used Fisher's identity to integrate out latent variables \cite{loaiza2024hybrid};
less directly, it is also related to a vast literature on latent variable-based, unbiased estimation of the likelihood and scalable Metropolis-Hastings with data subsampling \cite{nemeth2021stochastic,quiroz2019speeding,korattikara2014austerity}.
Most of these strategies, however, remain difficult to adapt to situations with complex random effect structures, mainly because the authors' focus is on models for which an \emph{unbiased} stochastic gradient is available, or the specific procedural steps of the Metropolis-Hastings algorithm.
We, on the other hand, consider \emph{biased} stochastic gradient estimators, which are easily derivable based on an additional layer of MCMC approximation.
The strategic control of gradient bias to reduce computational burden has been already explored in time series data analysis \cite{aicher2019stochastic,ou2026scalable,ou2026targeted}.
Our work furnishes experimental and theoretical evidence that similar strategy is also useful for hierarchical models.
~

~

To recapitulate, this paper combines two distinct strands of literature, about construction of the stochastic gradient suited to GLMMs and about the analysis of discretizing the underlying SDE with stochastic gradients, to propose a pragmatic and calibrated posterior computation method that is tailored to GLMMs.
Our contributions are threefold: (i) a Monte Carlo gradient estimator based on Fisher’s identity that enables SGMCMC in the presence of intractable marginal likelihoods; (ii) an asymptotic correction that adjusts for variance introduced by both minibatching and gradient approximation; and (iii) empirical results demonstrating accurate uncertainty quantification and substantial speedups over MCMC with the full dataset. Comparison of our method with the control variate SGLD designed for independent data show that failing to account for the additional noise from likelihood approximation leads to biased uncertainty estimates, underscoring the need for methods built specifically for GLMMs.

\section{Background on Diffusion and GLMMs}

We consider observed data $\mathbf{Y} = \left({Y}_1,\ldots,{Y}_n\right)$ of size $n \in \mathbb N$, where $Y_i$ represents the $i$-th observation for $i \in [n] := \{1, 2, \dots , n\}$. For population parameter $\boldsymbol{\Omega} \in \mathbb{R}^{d}$, the posterior distribution is $\pi\left(\boldsymbol{\Omega}\right) := p\left(\boldsymbol{\Omega} | \mathbf{Y}\right) \propto p\left(\boldsymbol{\Omega}\right) \prod_{i=1}^n p\left({Y}_i|\boldsymbol{\Omega}\right),$ where $p\left(\boldsymbol{\Omega}\right)$ is the prior distribution for the parameters $\boldsymbol{\Omega}$, $p\left({Y}_i|\boldsymbol{\Omega}\right)$ is the likelihood function and $\propto$ denotes proportionality up to a constant. 
For notational convenience, we define $f_i\left(\boldsymbol{\Omega}\right) = -\log p\left({Y}_i | \boldsymbol{\Omega}\right)$ for all $i \in [n]$,  $f_0\left(\boldsymbol{\Omega}\right) = -\log p\left(\boldsymbol{\Omega}\right)$ and $f\left(\boldsymbol{\Omega}\right) = \sum_{i=0}^n f_i\left(\boldsymbol{\Omega}\right)$; 
in particular, the posterior is then $\pi\left(\boldsymbol{\Omega}\right) \propto \exp\left\{-f\left(\boldsymbol{\Omega}\right)\right\}.$ We define the gradient of the marginal log-likelihood $g_i(\boldsymbol{\Omega}) = \nabla f_i(\boldsymbol{\Omega})$. It follows that the gradient of $f$ is $\nabla f(\boldsymbol{\Omega}) = \nabla f_0(\boldsymbol{\Omega}) + \sum_{i=1}^n g_i(\boldsymbol{\Omega})$. We will use the notation $h_{\mathcal S}(\boldsymbol{\Omega}) = \frac{1}{S} \sum_{i \in \mathcal S} g_{i}(\boldsymbol{\Omega})$, where $\mathcal S$ is a random subset of $[n]$ such that $|\mathcal S| = S$. 
The gradient of $f$ can then be written as $\nabla f (\boldsymbol{\Omega}) = \nabla f_0(\boldsymbol{\Omega}) + n h_{[n]}(\boldsymbol{\Omega})$.

We also introduce the following standard notations. For two sequences $\{a_n\}_n$, $\{g_n\}_n$, $g_n = O_n(a_n)$ means that there exists $C>0$ for which $\|g_n\| \leq C a_n$ as $n\to \infty$. We write $g_n = \Omega_n(a_n)$ if $1/g_n = O_n(1/a_n)$.
For a $d\times d$-dimensional real matrix $\mathbf M$, we define its operator norm as $\|\mathbf M\|_{op} = \sup_{\mathbf v\in \mathbb{R}^d : \|\mathbf v\|\leq1}\|\mathbf{Mv}\|$.
For higher-order tensors on $\mathbb{R}^d$, an operator norm is similarly defined: e.g., for a trilinear map $\mathbf{T}:\mathbb{R}^d\times\mathbb{R}^d\times\mathbb{R}^d\to\mathbb{R}$, we define $\|\mathbf{T}\|_{op} = \sup_{\mathbf{v}:\|\mathbf{v}\|\leq1}\mathbf{T}(\mathbf{v},\mathbf{v},\mathbf{v})$.

\subsection{Stochastic gradient Langevin dynamics}

The Langevin diffusion, $\boldsymbol{\Omega}_t$, is defined by the SDE,
\begin{equation}\label{e:sde0}
d\boldsymbol{\Omega}_t = -\nabla f\left(\boldsymbol{\Omega}_t\right)dt + \sqrt{2}d\mathbf{B}_t,\tag{SDE}
\end{equation}
where $\nabla f\left(\boldsymbol{\Omega}\right)$ denotes the gradient of $f$ in $\OO$ and $\mathbf{B}_t$ is a $d$-dimensional Brownian motion. The diffusion converges to the unique stationary distribution $\pi$ under mild conditions on $f$ \cite{roberts1996exponential}, in which case samples from $\pi$ can be obtained from stationary trajectories of \eqref{e:sde0}. 
In practice, Euler approximation can be used to discretize the stochastic process \eqref{e:sde0}.
With a fixed step size $\epsilon > 0$, the update scheme is given by
\begin{equation}\label{e:em0}
    \boldsymbol{\Omega}_{k+1} = \boldsymbol{\Omega}_k - \epsilon\nabla f\left(\boldsymbol{\Omega}_k\right) + \sqrt{2\epsilon}\boldsymbol{\eta}_k, \quad \boldsymbol{\eta}_k \stackrel{{iid}}{\sim} \mathcal N_d \left(\mathbf{0}_d, \mathbf{I}_d\right).
    \tag{Euler-Maruyama}
\end{equation}
The update of \eqref{e:em0} has a computational cost that scales linearly in $n$. In the regime  $n \gg 1$ (i.e., the big data regime of interest for this paper), this results in a prohibitive computational cost for the sampling process. 
To address the issue, Welling and Teh \cite{welling2011bayesian} introduced SGLD by replacing the full-data gradient $\nabla f\left(\boldsymbol{\Omega}\right)$ with an unbiased estimate $\nabla \hat{f}_{\rm SGLD}\left(\boldsymbol{\Omega}\right) = \nabla f_0\left(\boldsymbol{\Omega}\right) + \frac{n}{S}\sum_{i\in \mathcal S} \nabla  f_i\left(\boldsymbol{\Omega}\right) = \nabla f_0\left(\boldsymbol{\Omega}\right) + nh_{\mathcal S}(\OO)$, where  $\mathcal S$ is a random subset of $[n]$ drawn \emph{iid} at each iteration with $|\mathcal S| = S$. The minibatch estimator $h_{\mathcal S}(\OO)$ can easily be seen to be an unbiased estimator of the full-data gradient $h_{[n]}(\OO)$. When $S \ll n$, this significantly reduces the cost of the gradient computation. 
A single update of the \emph{idealized} SGLD algorithm is thus given by
\begin{equation}\label{e:sgld}
  \boldsymbol{\Omega}_{k+1} = \boldsymbol{\Omega}_k - \epsilon\nabla \hat{f}_{\rm SGLD}\left(\boldsymbol{\Omega}_k\right) + \sqrt{2\epsilon}\boldsymbol{\eta}_k;\quad
  \epsilon>0.
  \tag{SGLD}
\end{equation}
If the step size vanishes in the number of iterations, then \eqref{e:sgld} permits sampling from the posterior $\pi$ up to a demanded accuracy \cite{welling2011bayesian}. However, practical implementations of SGLD tend to take a fixed step size to avoid excessively slow convergence. 
This implies the resulting posterior has an inflated variance \citep{brosse2018promises}. 

Besides the inevitable variance inflation, implementation of SGLD poses a different problem for GLMMs; namely, the  intractability of the marginal log-likelihood.
This implies that the natural stochastic gradient estimator $h_{\mathcal S}(\OO)$ cannot be straightforwardly evaluated.
We now introduce the background for GLMMs and describe how SGLD must be adapted to fit them.

\subsection{Generalized linear mixed model}
\label{sec:glmm_background}

\begin{table}[t]
\centering
\caption{Examples of GLM observation likelihood specifications.}
\begin{tabular}{|l|c|c|c|c|} 
 \hline
 Model & Range of $Y$ & $\mu(\theta)=\mathbb{E}[Y|\theta]$ & Cumulant $b(\theta)$ & Canonical Link   \\ [0.5ex] 
 \hline\hline
 Gaussian & $\mathbb{R}$ & $\theta$ & 
 $\theta^2/2$ & Identity
 \\ 
 Binomial & $[M]$ & $Me^\theta/(1+e^\theta)$ & $M\log(1+e^\theta)$ & Logit
 \\
 Poisson & $\{0,1,2,\ldots,\}$ & $e^\theta$ & $e^\theta$ & Log
 \\
 \hline
\end{tabular}
\label{tab1:exp}
\end{table}

Suppose each group/subject $i\in [n]$ has $n_i \in \mathbb N$ repeated measures, $Y_{it}$ for $t \in [n_i]$. Let $\mathbf{x}_{it}$ denote a $p$-dimensional vector of covariates, and $\mathbf{z}_{it}$ denote a $q$-dimensional vector of covariates that are assumed to have subject-specific parameters. The elements of $\mathbf{Y}_i = (Y_{i1},\ldots,Y_{in_i})^\top$ are modeled as conditionally independent random variables from the exponential family,
\begin{equation} \label{e:exp_family}
p(Y_{it} | \theta_{it}, \phi) = \exp\left\{\frac{Y_{it} \theta_{it} - b(\theta_{it})}{a(\phi)} + c(Y_{it}, {\phi})\right\},
\end{equation}
where $\theta_{it}=\theta(\mu_{it})$ is a canonical parameter of an exponential family member that is a smooth function of the conditional mean function $\mu_{it} = \mathbb{E}[y_{it}|\mathbf{x}_{it},\mathbf{z}_{it}]$ (from hereon, notation for conditioning on regressor features is suppressed).
A known smooth link function $g(\cdot)$ relates $\mu_{it}$ to the linear predictor $\eta_{it} = \mathbf{x}_{it}^\top\boldsymbol{\beta} + \mathbf{z}_{it}^\top\boldsymbol{\gamma}_i$ through the relation $\eta_{it} = g(\mu_{it})$.
$\boldsymbol{\beta}$ is a $p$-dimensional vector of population regression parameters (i.e., \emph{fixed effects}), and $\boldsymbol{\gamma}_i$ for each $i\in[n]$ is a $q$-dimensional vector of subject-specific parameter deviations from $\boldsymbol{\beta}$ (i.e., \emph{random effects}). 
We assume the common normal random effects model: $\boldsymbol{\gamma}_i \stackrel{{iid}}{\sim} \mathcal{N}_q(0, \boldsymbol{\Sigma})$, though this can be somewhat relaxed, as long as the distributional model satisfies our abstract assumptions in Section \ref{sec:theory}. 
$\phi > 0$ is a dispersion parameter, and $a(\cdot), b(\cdot),c(\cdot)$ are known smooth real-valued functions, with $a>0$. We group the population parameters and write $\boldsymbol{\Omega} = (\boldsymbol{\beta}, \boldsymbol{\Sigma}, \phi)$.
Table \ref{tab1:exp} summarizes the models discussed in this paper.

A canonical link function $g(\cdot)$ corresponds to the choice of choosing $\theta_{it} = \eta_{it}$. However, there are situations where different link functions are used. An example is a probit model for $Y\in[M]$, for which we set $\Phi(\eta_{it}) = \mu_{it}$ ($\Phi$ the standard normal distribution function). 
This induces the relation $\theta_{it} = \log\frac{\Phi(\eta_{it})}{1-\Phi(\eta_{it})}$.
Similarly, link functions other than the logarithmic one can be chosen for a Poisson model.
The use of different link functions can make a noticeable difference in the applicability of our theoretical analyses and a simpler, unbiased design of a stochastic gradient estimator.
These points are discussed in more detail in Sections \ref{sec:practical} and \ref{sec:theory}.

The posterior distribution is denoted by $p(\boldsymbol{\Omega} | \mathbf{Y})$ with $\mathbf{Y} = (\mathbf{Y}_1^\top,\ldots,\mathbf{Y}_n^\top)^{\top}$. 
It is proportional to the marginal joint likelihood times the prior, $\pi(\boldsymbol{\Omega}) := p(\boldsymbol{\Omega} | \mathbf{Y}) \propto p(\mathbf{Y} | \boldsymbol{\Omega})p(\boldsymbol{\Omega})$. We assume that the prior and likelihood decompose as $p(\boldsymbol{\Omega}) = p(\boldsymbol{\beta})p(\boldsymbol{\Sigma})p(\phi)$ and $p(\mathbf{Y}|\boldsymbol{\Omega}) = \prod_{i=1}^n p(\mathbf{Y}_i | \boldsymbol{\Omega})$. The subject-specific likelihood contribution is given by the following $q$-dimensional integral that in general is intractable, 
\begingroup
\allowdisplaybreaks
\begin{align}\label{eq:ml}
p(\mathbf{Y}_i | \boldsymbol{\Omega}) = \int \prod_{t = 1}^{n_i} p(Y_{it} | \boldsymbol{\beta}, \boldsymbol{\gamma}_i, \phi) p(\boldsymbol{\gamma}_i | \boldsymbol{\Sigma}) d\boldsymbol{\gamma}_i.    
\end{align}
\endgroup
The integrand function corresponding to the likelihood $p(\mathbf{Y}_i,\boldsymbol{\gamma}_i|\OO)$ is referred to as the ``complete likelihood.''
The integral arising in GLMMs for binary/count data is analytically intractable, due to the non-Gaussian structure of the random effects or the nonlinearity of the link function. 
The resulting computational issue is well-noted in both frequentist and Bayesian literature. 
A popular, default implementation for estimating the maximum likelihood in GLMMs approximates the integral \eqref{eq:ml} using a single modal estimate and numerically estimates its gradient \cite{lme4}.
Our approach, instead, is to use Fisher's identity \citep{fisher1925theory} and explicitly target the gradient of the integral with respect to $\OO$, given by
\begin{equation}
g_i(\boldsymbol{\Omega}) =\mathbb{E}_{\boldsymbol{\gamma}_i| \mathbf{Y}_i, \boldsymbol{\Omega}}\left[-\nabla \log p(\mathbf{Y}_i , \boldsymbol{\gamma}_i | \boldsymbol{\Omega})\right],\label{e:fisher}
\end{equation}
where $g_i(\boldsymbol{\Omega}) = \nabla f_i(\boldsymbol{\Omega}) = -\nabla\log p(\mathbf{Y}_i | \boldsymbol{\Omega})$. 
Fisher's identity re-writes the gradient of the marginal log-likelihood as the expected gradient of the complete log-likelihood with respect to the conditional distribution $p(\boldsymbol{\gamma}_i| \mathbf{Y}_i, \boldsymbol{\Omega})$.
Fisher’s identity has a long history in models with latent variables, including linear dynamic systems \citep{segal1989new}, non-linear and non-Gaussian state space models \citep{kantas2015particle}, and high-dimensional inverse problems \citep{jasra2021unbiased}. \citet{tran2020bayesian} applied it to GLMMs in a variational inference framework incorporating neural networks. Their method estimated conditional log-likelihood gradients via backpropagation and approximated expectations using importance sampling. 

Our aim is not variational inference, but designing a posterior sampling method like Loaiza-Maya et al \cite{loaiza2024hybrid}.
We introduce a Monte Carlo estimator for $g_i(\boldsymbol{\Omega})$ based on Fisher's identity and incorporate it into an SGMCMC algorithm that permits a sampling-based Bayesian inference for large datasets. When used with independent Monte Carlo samples to approximate $g_i(\boldsymbol{\Omega})$, our approach yields unbiased stochastic gradients of the marginal log-likelihood, enabling SGLD updates for models with intractable marginal gradients.
More realistically, even when the samples are produced by a biased MCMC algorithm, we provide assumptions under which the theoretical analyses can provide rigorous guarantees on the sampler accuracy. Experimental results support the validity of the resulting approach.

\section{SGMCMC for GLMMs}

We now introduce our SGMCMC algorithm for scalable Bayesian inference in GLMMs. 
Our approach builds on Fisher’s identity to construct a Monte Carlo estimator of the gradient of the marginal log-likelihood, enabling the use of SGLD in this otherwise intractable setting. 
The SGLD algorithm \cite{welling2011bayesian} given in \eqref{e:sgld} approximates the full-data gradient, $\nabla f\left(\boldsymbol{\Omega}\right) = \nabla f_0(\boldsymbol{\Omega}) + n h_{[n]}(\boldsymbol{\Omega})$, with an unbiased estimator $\nabla \hat{f}_{\rm SGLD}\left(\boldsymbol{\Omega}\right) = \nabla f_0(\boldsymbol{\Omega}) + n h_{\mathcal S}(\boldsymbol{\Omega})$. The estimator from \cite{welling2011bayesian} relies on a vanishing step size and the fact that $\nabla  f_i\left(\boldsymbol{\Omega}\right)$ can be computed. In this section, we consider the following SGLD algorithm for GLMMs:
\begin{equation} \label{e:estimator}
\nabla \hat{f}_{\rm GLMM}\left(\boldsymbol{\Omega}
;\mathcal{S}
\right) = \nabla f_0\left(\boldsymbol{\Omega}\right) + \frac{n}{S}\sum_{i\in \mathcal S} \hat g_i\left(\boldsymbol{\Omega}\right) = \nabla f_0\left(\boldsymbol{\Omega}\right) + n \hat{h}_{\mathcal S}(\boldsymbol{\Omega}).
\end{equation}
Our estimator is a function of a ``population estimator'' $\hat{h}_{\mathcal S}(\boldsymbol{\Omega}) = \frac{1}{S}\sum_{i\in \mathcal S} \hat g_i(\boldsymbol{\Omega})$, which is a minibatch average of the subject-specific gradient estimators, $\hat g_i\left(\boldsymbol{\Omega}\right)$. Recall that $g_i(\boldsymbol{\Omega}) = \nabla f_i(\boldsymbol{\Omega}) = -\nabla\log p(\mathbf{Y}_i | \boldsymbol{\Omega})$ is the gradient of the marginal log-likelihood for subject $i$, and in the GLMM setting must be approximated. In Section \ref{sec:estimator}, we define an estimator for $g_i\left(\boldsymbol{\Omega}\right)$ and derive its statistical properties. We also characterize the population estimator $\hat{h}_{\mathcal S}(\boldsymbol{\Omega})$, show that it is an unbiased estimator of ${h}_{[n]}(\boldsymbol{\Omega})$, and introduce an estimator for its covariance. In Section \ref{sec:correction}, we derive a post-hoc correction for the posterior samples obtained from the algorithm in \eqref{e:estimator}.
In Section \ref{sec:practical}, we detail the practical implementation of the algorithm, providing a step-by-step procedure for inference.
A more detailed error analysis of the resulting algorithm and its post-hoc correction is postponed to Section \ref{sec:theory}. 

\subsection{Deriving an estimator for the marginal likelihood} \label{sec:estimator}

We begin by defining the estimator for $g_i(\boldsymbol{\Omega})$. Using \eqref{e:fisher}, we re-write $g_i(\boldsymbol{\Omega})$ as an expectation with respect to the conditional distribution of a random effect $\boldsymbol{\gamma}_i$. Equation \eqref{e:fisher} can be decomposed into two terms involving the conditional likelihood and the random effects density, respectively:
\begin{equation}
g_i(\boldsymbol{\Omega}) =\mathbb{E}_{\boldsymbol{\gamma}_i| \mathbf{Y}_i, \boldsymbol{\Omega}}\left[-\nabla \log p(\mathbf{Y}_i | \boldsymbol{\gamma}_i , \boldsymbol{\beta},\phi)-\nabla \log p(\boldsymbol{\gamma}_i | \boldsymbol{\Sigma})\right].\label{e:fisher2}
\end{equation}
Both the conditional likelihood, $ p(\mathbf{Y}_i | \boldsymbol{\gamma}_i , \boldsymbol{\beta},\phi)$, and the random effects distribution, $p(\boldsymbol{\gamma}_i | \boldsymbol{\Sigma})$, are members of the exponential family and are differentiable under mild conditions. The density of the random effects distribution is not required to be a member of the exponential family, but only to be continuously differentiable.

The representation of the gradient of the marginal log-likelihood as an expectation in \eqref{e:fisher2} motivates a Monte Carlo estimator for $g_i(\boldsymbol{\Omega})$ (see also \cite{tran2020bayesian}). 
We define our Monte Carlo estimator based on $R$ samples of $\boldsymbol{\gamma}_{i}$ as,
\begin{equation}\label{e:hatgi}
\hat g_i(\boldsymbol{\Omega}) = -\frac{1}{R}\sum_{r = 1}^R \nabla \log p(\mathbf{Y}_i, \boldsymbol{\gamma}_{ir} | \boldsymbol{\Omega}), \quad \boldsymbol{\gamma}_{ir} \stackrel{{iid}}{\sim} p(\boldsymbol{\gamma}_i| \mathbf{Y}_i, \boldsymbol{\Omega}).
\end{equation}
For the moment, let us assume that we can draw \emph{iid} Monte Carlo samples from the underlying conditional distribution $p(\boldsymbol{\gamma}_i|\mathbf{Y}_i,\boldsymbol{\Omega})$.
Compared with $\nabla\hat{f}_{\rm SGLD}$ appearing in \eqref{e:sgld}, this gradient estimator has higher variance due to additional Monte Carlo sampling.
The impact of using biased sampling methods or other approximations of $p(\boldsymbol{\gamma}_i|\mathbf{Y}_i,\boldsymbol{\Omega})$, and the corresponding posterior variance, will be addressed in Section \ref{sec:practical}.

The stochastic gradient estimator from \eqref{e:hatgi} is now straightforward to compute, since the gradient, $\nabla \log p(\mathbf{Y}_i, \boldsymbol{\gamma}_{ir} | \boldsymbol{\Omega})$, is analytically available. 
Previous analyses of SGLD algorithm \cite{mandt17,teh2016consistency,vollmer2016exploration} all point to the importance of the covariance of this Monte Carlo estimator \eqref{e:hatgi} as an additional contributor to the diffusive noise underlying the algorithm.
This covariance is given by
\begin{equation}\label{e:psii}
\boldsymbol{\Psi}_i(\boldsymbol{\Omega}) = \mathbb{V}_{\boldsymbol{\gamma}_i | \mathbf{Y}_i,\boldsymbol{\Omega}}\left(\hat g_i(\OO)\right) = \frac{1}{R}\mathbb{V}_{\boldsymbol{\gamma}_i | \mathbf{Y}_i,\boldsymbol{\Omega}}\left(\nabla \log p(\mathbf{Y}_i,\boldsymbol{\gamma}_i | \boldsymbol{\Omega})\right)\,,
\end{equation}
where $\mathbb V_{\mathbf X | \mathbf Y}(f(\mathbf X))$ is the conditional covariance of $f(\mathbf X)$ given $\mathbf Y$. The associated empirical covariance is given by
\begin{equation} \label{e:psiihat}
\hat{\boldsymbol{\Psi}}_i(\boldsymbol{\Omega}) = \frac{1}{R (R-1)}\sum_{r = 1}^R \left(\nabla \log p(\mathbf{Y}_i , \boldsymbol{\gamma}_{ir}| \boldsymbol{\Omega}) -\hat g_i(\boldsymbol{\Omega})\right)\left(\nabla \log p(\mathbf{Y}_i , \boldsymbol{\gamma}_{ir}| \boldsymbol{\Omega}) -\hat g_i(\boldsymbol{\Omega})\right)^\top\,.
\end{equation}
For each $i \in [n]$, $\hat g_i(\boldsymbol{\Omega})$ \eqref{e:hatgi} and $\hat{\boldsymbol{\Psi}}_i(\boldsymbol{\Omega})$ \eqref{e:psiihat} are unbiased estimators of $g_i(\boldsymbol{\Omega})$ \eqref{e:fisher2} and $\boldsymbol{\Psi}_i(\boldsymbol{\Omega})$ \eqref{e:psii}, respectively.
Based on the estimator of $g_i(\boldsymbol{\Omega})$, we can consider the statistical properties of the minibatch estimate of the \emph{population} gradient, $\hat h_{\mathcal S}(\boldsymbol{\Omega})$. 

\begin{lemma} \label{lemma:grad}
For every $\boldsymbol{\Omega} \in \mathbb R^d$, $\hat h_{\mathcal{S}}(\boldsymbol{\Omega})$ is an unbiased estimator of $h_{[n]}(\boldsymbol{\Omega})$ with covariance $S^{-1} \boldsymbol \Psi(\OO)$ and $\boldsymbol{\Psi}(\boldsymbol{\Omega}) = \frac{1}{n}\sum_{i=1}^n \left(g_i\left(\boldsymbol{\Omega}\right) - {h}_{[n]}(\boldsymbol{\Omega})\right)\left(g_i\left(\boldsymbol{\Omega}\right) - h_{[n]}(\boldsymbol{\Omega})\right)^\top + \frac{1}{n}\sum_{i=1}^{n}\boldsymbol{\Psi}_i(\boldsymbol{\Omega}).$
Furthermore,
\begin{align*}
\hat{\boldsymbol{\Psi}}(\boldsymbol{\Omega}) &:= \frac{1}{n}\sum_{i=1}^n\left\{\left(\hat{g}_i\left(\boldsymbol{\Omega}\right) - \hat{{h}}_{[n]}(\boldsymbol{\Omega})\right)\left(\hat{g}_i\left(\boldsymbol{\Omega}\right) - \hat{{h}}_{[n]}(\boldsymbol{\Omega})\right)^\top + \frac{\hat{\boldsymbol{\Psi}}_i(\boldsymbol{\Omega})}{n}\right\},\\
&= \frac{1}{n}\sum_{i=1}^n \left(\hat{g}_i\left(\boldsymbol{\Omega}\right) - \hat{{h}}_{[n]}(\boldsymbol{\Omega})\right)\left(\hat{g}_i\left(\boldsymbol{\Omega}\right) - \hat{{h}}_{[n]}(\boldsymbol{\Omega})\right)^\top + \frac{1}{n^2}\sum_{i=1}^n\hat{\boldsymbol{\Psi}}_i(\boldsymbol{\Omega}),
\end{align*}
is an unbiased estimator of $\boldsymbol{\Psi}(\boldsymbol{\Omega})$, where $\hat{{h}}_{[n]}(\boldsymbol{\Omega}) = \frac{1}{n}\sum_{i=1}^n \hat{g}_i\left(\boldsymbol{\Omega}\right)$.
\end{lemma}
\begin{proof}
See Section 1 of the online supplementary materials \cite{oursupplement}.
\end{proof}

The covariance of $\hat h_{\mathcal{S}}(\boldsymbol{\Omega})$ is given by $S^{-1}\boldsymbol{\Psi}(\boldsymbol{\Omega})$ and becomes smaller with increases in minibatch size ($S\rightarrow n$). 
It can be empirically estimated by $\hat{\boldsymbol{\Psi}}(\OO)$ as above. We decomposed it to reveal that for large $n$, contributions from group-specific, $R$-dependent Monte Carlo covariances \eqref{e:psii} become negligible.

Having defined our estimators, we can succinctly re-write the update equation as
\begin{align}
\label{e:oursgld}
\boldsymbol{\Omega}_{k+1} &= \boldsymbol{\Omega}_k - \epsilon \nabla \hat{f}_{\rm GLMM}(\boldsymbol{\Omega};
\mathcal{S}_k) + \sqrt{2 \epsilon} \boldsymbol{\eta}_k.
\tag{SGLD-GLMM}
\end{align}

\subsection{Variance correction for SGLD with an Extended Kalman Filter-Like Linearization} \label{sec:correction}

We can use the statistical characterization of the gradient estimators in Section \ref{sec:estimator} to rewrite our SGLD algorithm, exposing the noise added by approximations. By adding and subtracting the full-data gradient to \eqref{e:oursgld} we obtain,
\begin{align}
\boldsymbol{\Omega}_{k+1} = \boldsymbol{\Omega}_k-\epsilon \nabla f(\boldsymbol{\Omega}_k) + \sqrt{2\epsilon}\left( \sqrt{\frac{\epsilon n^2}{2S}}\boldsymbol{\xi}(\boldsymbol{\Omega}_k ; \mathcal{S}_k) + \boldsymbol{\eta}_k \right),\label{e:sde2}
\end{align}
where $\boldsymbol \xi(\boldsymbol{\Omega} ; \mathcal{S})$ is a set of random variables independent of $\boldsymbol \eta$ with mean zero and covariance $\mathbb V(\boldsymbol \xi(\boldsymbol{\Omega})) = \boldsymbol{\Psi}(\OO)$. The sum of the two independent variance terms is given by, $\boldsymbol{\Gamma}(\OO) = \epsilon n^2\boldsymbol{\Psi}(\boldsymbol{\Omega}) / (2S) + \mathbf{I}_d$. Note that \eqref{e:sde2} is identical to the Euler discretization given in \eqref{e:em0}, with the approximation noise isolated to the $\boldsymbol \xi(\boldsymbol{\Omega}; \mathcal{S})$ term. 
While the same observation was made by Mandt, Hoffman and Blei \cite{mandt17}, our attention here runs to a different direction: namely, the precise form of this covariance matrix provides a clue to variance correction.
Naively, we can impose either $\epsilon \approx n^{-2}$ or $S \approx n$, but these choices cancel out any computational benefits of subsampling the data. We hereby instead introduce an approach for correcting the posterior variance estimate based on asymptotic analysis.
The linearization technique we describe below is frequently used in applied SDE modeling literature and is essentially an extended Kalman filter method, which approximates the SDE with an appropriate Ornstein-Uhlenbeck process; see, e.g., \cite[Sections~9,~10.6--7]{sarkka2019applied}.
The asymptotic validity of the correction will be investigated in more detail in Section \ref{sec:theory}.

We start by approximating the discrete-time update in \eqref{e:sde2} with a forward Euler scheme update of matching mean and covariance,
\begin{equation}
\boldsymbol{\Omega}_{k+1}- \boldsymbol{\Omega}_k= -\epsilon \nabla f(\boldsymbol{\Omega}_k) + \sqrt{2\epsilon}\boldsymbol{\eta}'_k, \quad \boldsymbol{\eta}'_k \sim \mathcal N\left(\mathbf{0},\boldsymbol{\Gamma}(\OO_k)\right),
\end{equation}
where $\boldsymbol{\Gamma}(\OO) = \epsilon n^2\boldsymbol{\Psi}(\boldsymbol{\Omega}) / (2S) + \mathbf{I}_d$. The update approximates the evolution of the SDE, $d \boldsymbol{\Omega}_t = - \nabla f(\boldsymbol{\Omega}_t) dt + \sqrt 2 \sigma(\boldsymbol{\Omega}_t) d\mathbf{B}_t$ for $\sigma(\OO)\sigma(\OO)^\top = \boldsymbol{\Gamma}(\OO)$. By the law of large numbers, we expect $n^{-1}\nabla f(\OO) $ to converge to an $n$-independent function of $\OO$ so that $\nabla f(\OO) = \Omega_n(n)$ for $\OO \neq \OO^*$. Consequently, the distribution of $\OO_t$ will concentrate around the minimizer $\OO^*$ of $f$ justifying the linearization of the SDE approximation of the update around $ \boldsymbol{\Omega}^*$,
$$d (\OO_t - \OO^*) = - \nabla^2 f(\boldsymbol{\Omega}^*)(\OO_t - \OO^*) dt + \sqrt 2 \sigma(\OO^*) d\mathbf{B}_t\,.$$
The invariant measure of this linear equation can be explicitly computed as $\bar \pi(\OO) = Z^{-1} \exp\left\{ - (\OO - \OO^*)^\top\boldsymbol \Sigma^{-1}(\OO - \OO^*)/2\right\},$ where $\boldsymbol \Sigma$ satisfies the Lyapunov equation, 
\begin{equation}
\mathbf{A} \boldsymbol{\Sigma}  + \boldsymbol{\Sigma} \mathbf{A}^\top = 2\boldsymbol{\Gamma},
\label{e:2}
\end{equation}
for $\mathbf{A} = \nabla^2 f(\boldsymbol{\Omega}^*)$ and $\boldsymbol{\Gamma} = \epsilon n^2\boldsymbol{\Psi}(\boldsymbol{\Omega}^*) / (2S) + \mathbf{I}_d$. Since both $\boldsymbol \Sigma$ and $\mathbf A$ are symmetric,  \eqref{e:2} can be written as $\boldsymbol{\Sigma} \mathbf{A}  + \mathbf{A} \boldsymbol{\Sigma} = 2 \boldsymbol{\Gamma}$.
Solving this equation for $\mathbf{A}$, where $\boldsymbol\Sigma, \boldsymbol{\Gamma}$ are estimated from SGLD samples, yields an estimate of the desired covariance.

\begin{remark}
  Whenever $\epsilon/S = O_n(n^{-(2+\delta)})$ for a chosen $\delta>0$, contribution from a diffusive Gaussian per each iteration is negligible with respect to contribution from the variance of the stochastic gradient estimator. 
  This ``low variance'' regime can result from either $\epsilon \ll n^{-2}$ and $S$ being a fixed constant (tiny steps), or $\epsilon = O_n(n^{-(1+\delta)})$ and $1/S = O_n(1/n)$ (full data).
  In these two regimes, SGLD is expected to obey dynamics close to that of a fixed-step SGD.
  In the sequel, we will be mainly interested in the ``moderate variance'' regime where $\epsilon = O_n(n^{-(1+\delta)})$ and $S$ is a fixed constant in $n$.
  The choice of $\delta > 0$ will be necessary, as it can be shown that already for simpler models, the choice of $\epsilon = c/n$ for any fixed $c>0$ results in an irreducible variance of the stochastic gradient estimator; see \cite[Theorem~7]{brosse2018promises}.
\end{remark}

\subsection{ Practical implementation of the algorithm}
\label{sec:practical}

\subsubsection*{Description of the Correction Algorithm}

Using the SGLD estimator from Section \ref{sec:estimator} and the variance correction from Section \ref{sec:correction}, we describe how to obtain posterior samples from our algorithm with a proper variance. We begin by solving the Lyapunov equation in \eqref{e:2}. 
Recall that $\mathbf{A}^{-1} = \nabla^2 f(\OO^*)$ for the maximum likelihood estimator (MLE) $\OO^*$ and $\mathbf{\Gamma}$ represents the injected noise of the SDE at $\OO = \OO^*$. 
After running an SGLD sampler for say $k$ iterations, we can approximate the unknown quantity $\boldsymbol{\Sigma}$ using the MCMC sample covariance $\widehat{\boldsymbol{\Sigma}}$, and MLE $\OO^*$ with the MCMC sample mean $\widehat{\OO}^*$, which results in an estimate $\widehat{\mathbf{A}}$.
Writing Cholesky decompositions for the uncorrected and corrected covariance as, $\widehat{\boldsymbol{\Sigma}} = \mathbf{E}^\top \mathbf{E}$ and $\widehat{\mathbf{A}} = \mathbf{F}^\top \mathbf{F}$, the corrected samples from algorithm \eqref{e:oursgld} can be computed as
\begin{equation} \label{eq:correction}
\boldsymbol{\Theta}_k = 
\mathbf{G} \left(\boldsymbol{\Omega}_k - \widehat{\boldsymbol{\Omega}}^*\right) + \widehat{\boldsymbol{\Omega}}^*, \quad \mathbf{G} = \left(\mathbf{E}^\top \mathbf{F}\right)^{-1}.
\end{equation}
The full algorithm is detailed in Algorithm \ref{alg}.

To compute the stochastic gradient estimator, we need to evaluate the gradient of the complete log-likelihood, $\nabla \log p(\mathbf{Y}_i, \boldsymbol{\gamma}_{ir} | \boldsymbol{\Omega})$, with respect to $\OO$.
The gradient takes a simple form due to the common structure of GLM likelihoods described in Section \ref{sec:glmm_background}. 
For simplicity, we focus on the canonical link obtained by setting $\boldsymbol{\theta} = \boldsymbol{\eta}$.
The gradient with respect to $\boldsymbol{\beta}$ is then 
\begin{align} \label{eq:betagrad}
\nabla_{\boldsymbol{\beta}} \log p(\mathbf{Y}_i , \boldsymbol{\gamma}_{ir} | \boldsymbol{\Omega}) 
= \frac{\mathbf{X}_i^\top \left(\mathbf{Y}_i - \boldsymbol{\mu}_{ir}\right)}{a(\phi)},
\end{align}
where $\mathbf{X}_i = (\mathbf{x}_{i1},\ldots,\mathbf{x}_{in_i})^\top$ is an $n_i \times p$ dimensional matrix of covariates and $\boldsymbol{\mu}_{ir} = (\mu_{ir1}, \ldots, \mu_{irn_i})^\top$ is an $n_i$-dimensional vector. 
The mean function $\mu_{irt} = \mathbb{E}[y_{it}|\theta_{irt}]$ depends on $\boldsymbol{\gamma}_{ir}$ through the relation $\theta_{irt} = \mathbf{x}_{it}^\top\boldsymbol{\beta} + \mathbf{z}_{it}^\top \boldsymbol{\gamma}_{ir}$. 
\eqref{eq:betagrad} can be derived based on the properties of the cumulant generating function $b(\theta_{irt})$, namely, that its first two derivatives yield the mean and the variance of the distribution.
Similarly, the gradient with respect to the dispersion parameter is given by
\begingroup
\allowdisplaybreaks
\begin{align} \label{eq:dispersiongrad}
\nabla_{\phi} \log f(\mathbf{Y}_i , \boldsymbol{\gamma}_{ir}| \boldsymbol{\Omega}) 
= -\sum_{t = 1}^{n_i}\left[\frac{a'(\phi)\{Y_{it}\theta_{irt} - b(\theta_{irt})\}}{a(\phi)^2} + c'(Y_{it}, \phi)\right],
\end{align}
\endgroup
where $a'(\phi)$ and $c'(Y_{it},\phi)$ are derivatives with respect to $\phi$. 
The gradient for $\boldsymbol{\Sigma}$ reduces to $\nabla_{\boldsymbol{\Sigma}} \log f(\mathbf{Y}_i , \boldsymbol{\gamma}_{ir}| \boldsymbol{\Omega}) = \nabla_{\boldsymbol{\Sigma}} \log f(\boldsymbol{\gamma}_{ir} | \boldsymbol{\Sigma})$.
Explicit expressions of the gradients required for handling binomial, Gaussian and Poisson regressions are derived in Sections 3--4 of the online supplementary materials \cite{oursupplement}.

\begin{algorithm}[t]
\caption{SGLD with Covariance Correction for GLMM}\label{alg}
\begin{algorithmic}
\Require{$\OO_0, S, \delta, K, R,$}
\State Define $\epsilon = S/n^{1+\delta}$
\For{$k \in 1,\ldots,K$}
    \State Draw $\mathcal{S}_{k} \sim \mathcal S$
    \For{$i \in \mathcal S_k$}
        \State Initialize: $\boldsymbol{\gamma}_{i0}$
        \For{$r \in 1,\ldots,R$}
        \State Sample $\boldsymbol{\gamma}_{ir}$ from $p(\boldsymbol{\gamma}_{i} | \mathbf{Y}_i,\boldsymbol{\Omega}_k)$
        \State Compute $\nabla_{\boldsymbol{\Omega}} \log p(\mathbf{Y}_i,\boldsymbol{\gamma}_{ir} | \boldsymbol{\Omega}_k)$
        \EndFor
        \State Compute $\hat g_i(\boldsymbol{\Omega}_k) = \frac{1}{R} \sum_{r=1}^R \nabla_{\boldsymbol{\Omega}} \log p(\mathbf{Y}_i,\boldsymbol{\gamma}_{ir} | \boldsymbol{\Omega}_k)$
        \Comment{See also \eqref{e:mcmcgi}}
    \EndFor
    \State Compute $\hat{h}_{\mathcal S_k}(\boldsymbol{\Omega}_k) = \frac{1}{S} \sum_{i \in \mathcal S_k} \hat g_i(\boldsymbol{\Omega}_k)$
    \State Draw $\boldsymbol{\eta}_k \sim N(0, 2\epsilon \mathbf{I}_d)$
    \State $\OO_{k+1} \gets \OO_k - \epsilon \left(\nabla f_0(\OO_k) + n\hat{h}_{\mathcal{S}_{k}}(\OO_k)\right) +  \boldsymbol{\eta}_k$
\EndFor
\State Compute posterior mean of unconstrained parameters, $\widehat{\boldsymbol{\Omega}}^*$
\State Compute correction, $\hat{\boldsymbol{\Psi}}(\widehat{\boldsymbol{\Omega}}^*)$ \Comment{Computation for $\hat{\boldsymbol{\Psi}}(\boldsymbol{\Omega})$ is given in Lemma \ref{lemma:grad}}
\State Compute injected noise, $\boldsymbol{\Gamma} = \boldsymbol{\Gamma}(\widehat{\boldsymbol{\Omega}}^*) = \epsilon n^2\hat {\boldsymbol{\Psi}}(\boldsymbol{\Omega}^*) / (2S) + \mathbf{I}_d$
\State Compute inflated posterior variance of $\boldsymbol{\Omega}$, $\widehat{\boldsymbol{\Sigma}}$
\State Solve Lyapunov equation to find corrected posterior variance, $\widehat{\mathbf{A}}^{-1}$
\State Obtain corrected samples $\boldsymbol{\Theta}_k$ \Comment{Definition of $\boldsymbol{\Theta}_k$ is given in \eqref{eq:correction}}
\end{algorithmic}
\end{algorithm}

\subsubsection*{Biased Stochastic Gradients}

Having provided an implementable corrected sampling algorithm, we now turn to the issue of practical feasibility in constructing the stochastic gradient \eqref{e:hatgi}.
Recall that we assumed the feasibility of \emph{iid} Monte Carlo sampling from the conditional distribution $p(\boldsymbol{\gamma}_i|\mathbf{Y}_i,\boldsymbol{\Omega})$. 
The key here is that doing so results in both an unbiased estimator of the population gradient $h_{[n]}(\boldsymbol{\Omega})$ and a known covariance formula, as shown in Lemma \ref{lemma:grad}.
Situations in which this is possible are limited; as mentioned in Section \ref{sec:glmm_background}, the specific choice of link function can make a difference.
There are two known situations in which we can produce an unbiased stochastic gradient estimator \eqref{e:hatgi} using \emph{iid} Monte Carlo samples:
linear mixed models (LMMs) and probit binary regression model ($M=1$ in Table \ref{tab1:exp}).
The methods developed in this work will be mostly extraneous for the former case, for which an analytic integration in \eqref{eq:ml} is feasible.
In the latter case, results of Durante \cite{durante2019conjugate} can be used to show that the density of $p(\boldsymbol{\gamma}_i|\mathbf{Y}_i,\boldsymbol{\Omega})$ belongs to a family of unified skew-normal distributions \cite{arellano2006unified}.
A rejection sampler for a high-dimensional truncated normal distributions \cite{botev2017normal} can thus be used to draw \emph{iid} samples from the conditional distribution $p(\boldsymbol{\gamma}_i|\mathbf{Y}_i,\OO)$, for each $i\in[n]$.
    Botev \cite{botev2017normal} suggested that the numerical stability of the sampler begins to deteriorate around a dimension of 100.
    In our context, the worst-case dimension for using such a rejection sampler within our algorithm is $\max_{i=1,2,\ldots,n} n_i$, the largest number of within-subject measurements.
    Therefore, producing \emph{iid} random effects samples in this latter case is a genuine possibility only if there are very many groups in the dataset, but few groups have very many measurements.

For all other situations, including the logistic binary regression, there are currently no general recipes to draw  \emph{iid} samples from the conditional distribution $p(\boldsymbol{\gamma}_i|\mathbf{Y}_i,\boldsymbol{\Omega})$. 
We thus consider modifications to \eqref{e:hatgi} that are no longer unbiased Monte Carlo estimators. 
A straightforward extension is to use a biased stochastic gradient estimator that replaces an \emph{iid} sample average with an ergodic average with respect to $S$ Markov chains, each of whose stationary distribution is $p(\boldsymbol{\gamma}_i|\mathbf{Y}_i,\OO)$.
Write by $\mathcal{P}_i$ its transition kernel, and suppose the initial distribution is $\rho_{0i}$.
Then, our new stochastic gradient can be written as
\begin{equation}
    \label{e:mcmcgi}
    \hat g_i^{\rm MCMC}(\boldsymbol{\Omega}) = -\frac{1}{R}\sum_{r = 1}^R \nabla \log p(\mathbf{Y}_i, \boldsymbol{\gamma}_{ir} | \boldsymbol{\Omega}), \quad \boldsymbol{\gamma}_{ir}~{\sim}~
    \mathcal{P}^r_i\rho_{0i}(\boldsymbol{\gamma}_i| \mathbf{Y}_i, \boldsymbol{\Omega})~\forall r\in[R].
\end{equation}
We use \eqref{e:mcmcgi} instead of \eqref{e:hatgi} in the described SGLD algorithm \eqref{e:oursgld}.
Since this new algorithm relies on an additional MCMC iteration within each iteration, we call this a ``double loop'' algorithm.
The double loop version $\hat{g}_i(\OO)$ is no longer unbiased, so the covariance formula derived in Lemma \ref{lemma:grad} is technically incorrect.
The theoretical considerations in the following section are meant to address possible accuracy concerns with this MCMC-based implementation.
While, intuitively, the use of MCMC samples inflates the variance of the stochastic gradient, we will study in Section \ref{sec:theory} to what extent such an inflation factor is essentially negligible, and whether the correction in Algorithm \ref{alg} still works well.

\subsubsection*{Tuning Burden}

The SGLD correction proposed in this section relies on an asymptotic justification (the ``large $n$'' regime) and a proper choice of several tuning parameters. 
In Section \ref{sec:asymp}, we embark on a more thorough analysis that illustrates the effect of potentially large parameter space dimension $d$.
Notwithstanding the theoretical justification, large $d$ also imposes a  computational burden, as solving a general Lyapunov equation of form in general costs \eqref{e:2} $O_d(d^3)$ operations.
Furthermore, several important tuning parameters, including $S$ (the size of subsampled data), $\delta$ (dictating the step size), $R$ (length of the inner loop MCMC chain), and additional tuning parameters of the inner loop MCMC algorithm, all have an intrinsic dependence on $d$.
Tuning the first two parameters is a general problem for SGLD-type algorithms, for which previous authors have proposed a bandit-based algorithm \citep{coullon2023efficient}.
More research is needed to test their method in various model scenarios, though our results in Sections \ref{sims} and \ref{sec:da} suggest partial robustness of the results obtained with different choice of $S$.
We also suggest heuristic decision rules for choosing $\delta$ therein.

Tuning $R$ and other inner loop chain parameters is also important, though this may be automated to an extent for common GLMMs, such as the logistic/Poisson regression models. 
Proposition \ref{prop:ergodic} in Section \ref{sec:theory} and simulation results for logistic GLMM in Section \ref{sims} suggest that if there is an easy data augmentation scheme readily implementable with good mixing properties \cite{choi2013polya}, 
then a moderately large choice of $R$ (say 1,000) can provide accurate enough a posterior approximation for most tasks.
This is not true for models like Poisson GLMMs, in which case less complicated samplers may be used at the cost of increased $R$ for stability.
Moreover, to minimize tuning burden, a ``warm start'' initialization and adaptive sampling methods must be favored in the inner loop. 
In the discussion, we mention ongoing investigations of alternative strategies to lessen the non-negligible tuning burden imposed by MCMC-based stochastic gradient estimation.

\section{Theoretical Considerations}
\label{sec:theory}

In this section, we discuss questions about whether theory can support the empirically promising results of the proposed method to be provided in Sections \ref{sims}--\ref{sec:da}.
We will first analyze the \emph{uncorrected} SGLD sampler's error with respect to the true posterior, before moving onto the asymptotic analysis of the validity of our correction formula.

\subsection{Sampling error analysis}

We first make the following assumptions about $f$.
Note that they imply: (1) $\nabla f$ is $(n+1)L$-Lipschitz continuous; and, (2) the existence of a unique minimizer $\OO^*$ of the negative log-posterior density $f$ for each $n$.
\begin{assumption}\label{a:smooth}
    Each $f_i,i\in\{0\}\cup[n]$ is twice continuously differentiable.
    There is a positive constant $L$ such that $\|\nabla f_i(\OO) - \nabla f_i(\OO')\|\leq L\|\OO - \OO'\|$ for every $i\in\{0\}\cup[n]$ and every $\OO,\OO'\in\mathbb{R}^d$.
\end{assumption}

\begin{assumption}\label{a:convex}
    Each $f_i,\;i\in\{0\}\cup[n]$ is convex, and there is a positive sequence $m = m_{n,d}$ such that $\langle\OO-\OO',\; \nabla f(\OO) - \nabla f(\OO')\rangle
    \geq m\|\OO - \OO'\|^2$ for every $n$ and $\OO,\OO'\in\mathbb{R}^d$.
\end{assumption}

\begin{remark}
Assumptions \ref{a:smooth} and \ref{a:convex} are not satisfied by many marginal GLMM likelihoods.
Nevertheless, the posterior of $\OO$ is often locally Lipschitz/strongly convex over any bounded sets. 
In particular, treating $m$ to grow like $d/n$ on some bounded set is sensible in light of known posterior contraction results, in which case a conclusive analysis will need to bound the exit time of the underlying SGLD Markov chain within the number of iterations that scale at best polynomially with $n$.
Such an analysis was previously carried out by Altmeyer \cite{altmeyer2022polynomial} for GLMs and Zhang and Li \cite{zhang2025pigeonhole} for crossed effects linear models.
In Section \ref{sims}, we will empirically investigate the practical impact of violations of assumptions \ref{a:smooth} and \ref{a:convex}.
\end{remark}

Our main result about the uncorrected SGLD sampler is now as follows.

\begin{theorem}\label{thm:conv}
    Write by $\rho_k$ the distribution of $\OO_k$ at the $k$-th iteration of SGLD update \eqref{e:oursgld} that uses the stochastic gradient \eqref{e:mcmcgi}. 
    Under assumptions \ref{a:smooth} and \ref{a:convex}, suppose there also exists constants $c=c_{R,d}$, $B=B_{R,d}$ such that for every $i\in[n]$ and $k=1,2,\ldots$ we have:
    \begin{equation}\label{eq:unif_bias}
        \mathbb{E}[\|\mathbb{E}[\hat{g}_i^{\rm MCMC}(\OO_k)|\OO_k] - \nabla f_i(\OO_k)\|^2]\leq c^2d,
    \end{equation}
    \begin{equation}\label{eq:unif_var}
        \mathbb{E}[\|\hat{g}_i^{\rm MCMC}(\OO_k) - \mathbb{E}[\hat{g}_i^{\rm MCMC}(\OO_k)|\OO_k]\|^2|\OO_k] \leq Bd.
    \end{equation}
    Then, for $\epsilon\leq\frac{2}{(n+1)L+m}$, we have
    \begin{equation*}
        W_2(\rho_k,\pi) \leq (1-m\epsilon)^kW_2(\rho_0,\pi) + 1.65\frac{L}{m}\sqrt{\epsilon d} + \frac{nc\sqrt{d}}{m} + \frac{n^2B\sqrt{\epsilon d}}{(1.65LS + n\sqrt{SBm})}.
    \end{equation*}
\end{theorem}

\begin{proof}
    See Section 2 of the online supplementary materials \cite{oursupplement}.
\end{proof}

In the above bound, taking $m=\Omega_n(n)$ reveals that the fourth term is $O_n(\sqrt{n\epsilon}) = o_n(1)$ for fixed $B,S,d$, whenever $\epsilon=o_n(n^{-1})$.
The third term, on the other hand, cannot be made arbitrarily small by either large $n$ or small $\epsilon$, due to the irreducible bias of the stochastic gradient. 
However, based on the ergodic theorem, we again expect that $c\sim 1/R$.
Although the absolute constants are unknown, this suggests that the effect of biased gradient is negligible for reasonable choices of $n,\epsilon,R$ (e.g., $n=10^5$ and $\epsilon=n^{1.1}$ yields $\sqrt{n\epsilon} \approx 0.56$).
The following proposition shows that this is true under a few stringent regularity assumptions.

\begin{proposition} \label{prop:ergodic}
    Suppose a family of functions $\boldsymbol{\gamma}\mapsto \nabla\log p(\mathbf{Y}_i,\boldsymbol{\gamma}|\OO)$ from \eqref{e:hatgi} is $L = L_{d}$-Lipschitz continuous, uniformly in $i$ and $\OO$. 
    Suppose also that a family of distributions $\rho_{0i}$ and Markov chain transition kernels $\mathcal{P}_i$ in \eqref{e:hatgi}, for every $i$ and $\OO$, satisfy a uniform 2-Wasserstein geometric ergodicity condition:
    \begin{equation} \label{e:geometric}
        W_2(p(\boldsymbol{\gamma}_i|\mathbf{Y}_i,\OO),\; \mathcal{P}_i^r\rho_{0i}(\boldsymbol{\gamma}_i|\mathbf{Y}_i,\OO)) \leq \lambda^r E_0.
    \end{equation}
    Then, \eqref{eq:unif_bias} is satisfied with
    $c \leq \frac{LE_0}{R(1-\lambda)}$.
\end{proposition}

\begin{proof}
        See Section 3 of the online supplementary materials \cite{oursupplement}.
\end{proof}

\begin{remark}
There is a trade-off between running the Markov chains longer to decrease the bias, and saving computational expense with shorter chains.
Practical recommendations for choosing $R$ may be based on an optimal finite-sample bounds, which is beyond the scope of the current paper.
Nevertheless, the user has some control over adapting the initialization $\rho_{0i}$ for each MCMC sampler targeting the conditional posterior of $\boldsymbol{\gamma}_i$, rendering the geometric ergodicity assumption with a uniform constant $E_0$ more plausible. 
We recommend initializing the Markov chain targeting the conditional posterior of $\boldsymbol{\gamma}_i$ near the mode of $\boldsymbol{\gamma}_i\mapsto\log p(\mathbf{Y}_i,\boldsymbol{\gamma}_i|\OO)$, for each $(i,\OO)$.
Approximating the mode of $\boldsymbol{\gamma}_i$ can be automated using off-the-shelf algorithms, such as a simple Fisher scoring method with adaptively decreasing step size to prevent divergence \cite{glm2}.
Empirical investigations of sensitivity of the posterior variance estimate to the choice of $R$ are carried out in Sections 6.2 and 6.3 of the online supplementary materials.
\end{remark}

\subsection{Asymptotics of correction} \label{sec:asymp}

Next, we turn to asymptotic analyses that support the heuristic insight behind the Lyapunov equation \eqref{e:2}, as explained in Section \ref{sec:correction}, which is the basis of our correction method in Algorithm \ref{alg}.
Recall the exact expression of the covariance matrix $\boldsymbol\Psi(\OO)$ of GLMM-specific stochastic gradient estimator $\hat{h}_{\mathcal S}(\OO)$ given in Lemma \ref{lemma:grad}.
This expression is no longer correct for an MCMC-based gradient estimator $\hat{g}_i^{\rm MCMC}(\OO)$ \eqref{e:mcmcgi}, due to the dependence within each $R$ samples of $\boldsymbol{\gamma}_{ir}$s.
However, based on the ergodic central limit theorem, we would expect its covariance to be similar to that given by Lemma \ref{lemma:grad} for large $R$. 
Of course, the covariance $\boldsymbol\Psi(\OO)$ itself depends implicity on $R$, the number of independent Monte Carlo samples. 
More specifically, the covariance depends on an average of $\hat{\boldsymbol{\Psi}}_i(\OO)/n$ over $i$ \eqref{e:psiihat}, which we expect to be negligible for either large $n$ or $R$.

The foregoing suggests that for theoretical purposes, it is better to analyze to what extent the Lyapunov equation 
\begin{equation} \label{e:3}
    \mathbf A\tilde{\boldsymbol \Sigma}_k +  \tilde{\boldsymbol \Sigma}_k \mathbf A = 2\boldsymbol\Gamma_0
\end{equation}
is true, rather than \eqref{e:2}, where
$\tilde{\boldsymbol\Sigma}_k$ is the covariance of the $k$-th SGLD iterate $\OO_k\sim\rho_k$, and
\begin{align*}
    \boldsymbol{\Gamma}_0(\OO) &= \frac{\epsilon n^2}{2S}\boldsymbol{\Psi}_0(\boldsymbol{\Omega})+ \mathbf{I}_d,\\
    \boldsymbol\Psi_0(\OO) &= \frac{1}{n}\sum_{i=1}^{n}\{g_i(\OO)-h_{[n]}(\OO)\}\{g_i(\OO)-h_{[n]}(\OO)\}^\top.
\end{align*}
It is easy to see that $\lim_{R\to\infty}\boldsymbol{\Gamma}(\OO) = \boldsymbol\Gamma_0(\OO)$ for each $\OO$, 
and that $\boldsymbol\Psi_0(\OO)$ is exactly the covariance matrix of an \emph{idealized} algorithm \eqref{e:sgld}, where we have $\nabla\hat{f}_{\rm SGLD}(\OO) = \frac{n}{S}\sum_{i\in\mathcal S}g_i(\OO).$
In fact, it is not necessary to directly tackle \eqref{e:3}.
Provided the chain was run sufficiently long enough with large enough $R$, one can only study the extent to which \eqref{e:3} is true with respect to the invariant distribution of an idealized SGLD algorithm, which uses $g_i(\OO)$ rather than $\hat{g}_i^{\rm MCMC}(\OO)$ to compute the stochastic gradient estimator.
To see this, observe that 
\begin{equation*}
    \|\tilde{\boldsymbol\Sigma}_k\mathbf A + \mathbf A\tilde{\boldsymbol\Sigma}_k - \boldsymbol\Sigma\mathbf A - \mathbf A\boldsymbol\Sigma\|_{op} \leq 2\|\mathbf A\|_{op}\|\tilde{\boldsymbol\Sigma}_k - \boldsymbol\Sigma\|_{op},
\end{equation*}
where $\|\mathbf A\|_{op}\leq (n+1)L$ by assumption \ref{a:smooth}, and $\boldsymbol\Sigma$ is the covariance matrix of the invariant distribution of SGLD algorithm \eqref{e:sgld}.
One can further split up the covariance difference into $\|\tilde{\boldsymbol \Sigma}_k - \boldsymbol{\Sigma}_0\|_{op} + \|\boldsymbol\Sigma - \boldsymbol\Sigma_0\|_{op}$, where $\boldsymbol\Sigma_0$ is the true posterior covariance.
The following result comes from \cite[p.24]{kasprzak2025good}.

\begin{proposition}
    For two probability distributions $P,Q$ on $\mathbb{R}^d$ which has mean zero and finite covariance matrices $C_1,C_2$, we have
    \begin{equation*}
        \|C_1 - C_2\|_{op}\leq W_2^2(P,Q) + 2W_2(P,Q)\sqrt{{\rm Tr}(C_2)}.
    \end{equation*}
\end{proposition}

Thus, bounds on the 2-Wasserstein distances between the posterior and the SGLD algorithm iterates and that between the posterior and the invariant distribution of SGLD per \eqref{e:sgld} can be used to control $\|\tilde{\boldsymbol \Sigma}_k - \boldsymbol{\Sigma}_0\|_{op} + \|\boldsymbol\Sigma - \boldsymbol\Sigma_0\|_{op}$.
One such result, for the distance between our SGLD algorithm and the posterior, is given by Theorem \ref{thm:conv} above; the other distance can be controlled similarly (in fact, more easily) as in the proof of Theorem \ref{thm:conv}, under assumptions \ref{a:smooth}--\ref{a:convex}, which are directly made on the gradients of marginal likelihoods.

Let us now focus on the validity of Lypaunov equation \eqref{e:3} for the idealized SGLD algorithm \eqref{e:sgld} using gradient estimator $\nabla\hat{f}_{\rm SGLD}$.
To simplify the bound, we make certain assumptions that are realistic in practice.

\begin{theorem}\label{t:main2} 
Let $\mathbf A  := \nabla^2 f(\OO^*)$ and set $\epsilon \leq \min\{\frac{1}{24n^{1+\delta}L},\frac{1}{n^{1+\delta}d}\}$ for some $\delta\in(0,1)$ such that $n^{1+\delta}\geq(1/2L)$.
In addition to assumptions \ref{a:smooth}--\ref{a:convex}, suppose:
\begin{itemize}
    \item $m$ from assumption \ref{a:convex} satisfies $m\geq \alpha n$ for some $\alpha = \alpha_d > 0$;
    \item each $f_i,\; i\in\{0\}\cup[n]$ is four times continuously differentiable with derivatives bounded by $L = L_d$ from assumption \ref{a:smooth};
    \item $\|\boldsymbol{\Psi}_0(\OO^*)\|_{op}$ is $O_n(1)$ and $\Omega_n(1)$; and,
    \item $\max_{i\in[n]}\|g_i(\OO^*) - h_{[n]}(\OO^*)\|^2 = O_n(d)$.
\end{itemize}
Then, we have
\begin{align*}
&\|{\boldsymbol \Sigma} \mathbf{A} + \mathbf{A} {\boldsymbol \Sigma} - 2\boldsymbol{\Gamma}_0\|_{op} = O_n\bigg(\frac{L}{\sqrt{\alpha}}\max\{n^{1/2-\delta},n^{1-3\delta/2}/\sqrt{d}\} + \frac{L^2}{\alpha}\max\{1/n^{\delta},n^{1-2\delta}/d\}\\
&\quad+ \frac{L}{\min\{\alpha,\sqrt{\alpha}\}}\max\{\sqrt{d}/n^\delta,n^{1-2\delta}/\sqrt{d},n^{1/2-2\delta}\} \\
&\quad+\frac{L}{\min\{\alpha,\alpha^2\}}\max\{d/n^{1+\delta},n^{1-3\delta}/d,1/n^{3\delta}\}\\
&\quad+\frac{L(L+1)}{\min\{\alpha,\alpha^{3/2}\}}\max\{n^{1/2-\delta}/\sqrt{d},n^{1-3\delta/2}/d\}\max\{d/n,1/n^{\delta},\sqrt{d}/n^{1/2+\delta}\}\bigg),
\end{align*}
where $\boldsymbol{\Sigma}$ is the covariance of the stationary distribution of idealized SGLD algorithm \eqref{e:sgld}.
\end{theorem}

\begin{proof}
    See Section 4 of the online supplementary materials \cite{oursupplement}.
\end{proof}

\begin{remark}
    The third and fourth additional assumptions of this theorem are not necessary for an explicit bound, which is to be found in the proof;
    however, they simplify the bound so that it is more interpretable in terms of $n,d$ and $\alpha$.
    To make sense of the bound, note that under the stated assumptions, the operator norm of $2\boldsymbol{\Gamma}_0$ can be chosen to be a constant multiple of $n^{1-\delta}$ for large enough $n$.
    If $d$ (and thus $\alpha, L$) are fixed constants, therefore, the error bound of the correction can be shown to be $\|\boldsymbol\Gamma_0\|_{op}$ times $O_n(\sqrt{n\epsilon}) = O_n(n^{-\delta/2})$ term.
    In genuinely high-dimensional problems, the growth rate of $L/\sqrt{\alpha}$ in $d$ cannot be ignored. 
    If $L/\sqrt{\alpha}$ grows like $d^\rho$ for some power $\rho< 1$, one can still obtain the same relative error, as long as $d = o_n(\sqrt{n})$ and $\delta$ is chosen to be small enough.
    The $d\ll\sqrt{n}$ growth condition is in fact known to be necessary for general accuracy results for the Laplace approximation \cite{kasprzak2025good}, which is closely related to the validity of the linearization of the underlying diffusion process.
\end{remark}

One thing the above results do \emph{not} address is the accuracy of the Laplace approximation for the posterior, based on the matrix $\mathbf A$.
Under assumptions \ref{a:smooth} and \ref{a:convex}, \cite[Proposition~5]{brosse2018promises} is applicable and yields a standard $O_n(n^{-1})$ error rate.
For \emph{iid.} models, very comprehensive results exist on asymptotic distributional properties of Bayes posteriors (see, e.g., \cite{kasprzak2025good}); in contrast, GLMMs have not received nearly as much attention.
Proving the validity of Laplace approximation for mixed models goes beyond the scope of this work. We do mention that a fairly rigorous, finite-sample analysis seems within the scope of existing tools based on the recent study \cite{katsevich2026highdimensionallaplaceasymptoticsconcentration}.

\section{Simulations}
\label{sims}

In this section, we demonstrate the performance of Algorithm \ref{alg} in various settings. We show that our algorithm properly estimates the posterior mean and variance and is a viable technique when the sample size becomes large, in contrast to existing methods, including Gibbs sampling and a control-variate SGLD.

\subsection{Illustration: LMM with Known Variance Components}
\label{sims:fixed}

We first investigated the performance of Algorithm \ref{alg} in estimating various quantities of the posterior distribution. 
Data was generated from a linear mixed model, ${Y}_{it} \stackrel{{ind}}{\sim} \mathcal N(\mathbf{x}_{it} \boldsymbol{\beta} + \mathbf{x}_{it} \boldsymbol{\gamma}_i, \sigma^2)$ and $\boldsymbol{\gamma}_i | \boldsymbol{\Sigma} \stackrel{{iid}}{\sim} \mathcal N(\mathbf{0},\boldsymbol{\Sigma})$.
Since the likelihood is Gaussian, the following quantities can be found in closed form: posterior for $\boldsymbol{\beta}$, gradient of the marginal log-likelihood, and posterior predictive distribution (PPD). 
We assumed $\mathbf{x}_{it} = \left(1,x_{it}\right)^\top$ and $x_{it} \stackrel{{iid}}{\sim} \mathcal N(0,1)$ and true parameters, $\boldsymbol{\beta} = (1.5, -0.5)^\top$, $\sigma^2 = 2$ and $\boldsymbol{\Sigma} = \begin{pmatrix} 1.5 & -0.25\\ -0.25 & 1.5 \end{pmatrix}$. We analyzed sample sizes of $n \in \{10^2, 10^3\}$, and set $n_i = 10$. For each value of $n$, we simulated 100 data sets. 

We used Algorithm \ref{alg} to sample from the posterior of $\boldsymbol{\beta}$ and present results for both the corrected ($\boldsymbol{\Theta}$) and uncorrected ($\boldsymbol{\Omega}$) posterior samples.
Importantly, in this experiment, we chose to fix $\boldsymbol{\Sigma}$ and $\sigma^2$ as \emph{known, oracle} values, and focus on the algorithm's ability to infer correct posterior mean and variance for the coefficient parameters.
This implies that the experiment mainly serves a pedagogic purpose; it avails us to directly evaluate the algorithm's ability to recover the true posterior distribution, the gradient of the marginal log-likelihood, and the PPD.
Weakly informative, independent $\mathcal N(0, 10^2)$ priors were used for $\beta_0,\beta_1$. We present results for $S \in \{1, 5, 10\}$ and $\delta = \{0.1,0.2,\ldots,1\}\cap\{\delta: \epsilon := S/n^{1+\delta} < n^{-1}\}$. $R$ was set to 100 throughout all simulations. The algorithm was run for 100 continuous time steps (i.e., $100 / \epsilon$) and was thinned to 5,000 samples.

\begin{figure}[t]
\begin{center}
\includegraphics[width=0.99\textwidth]{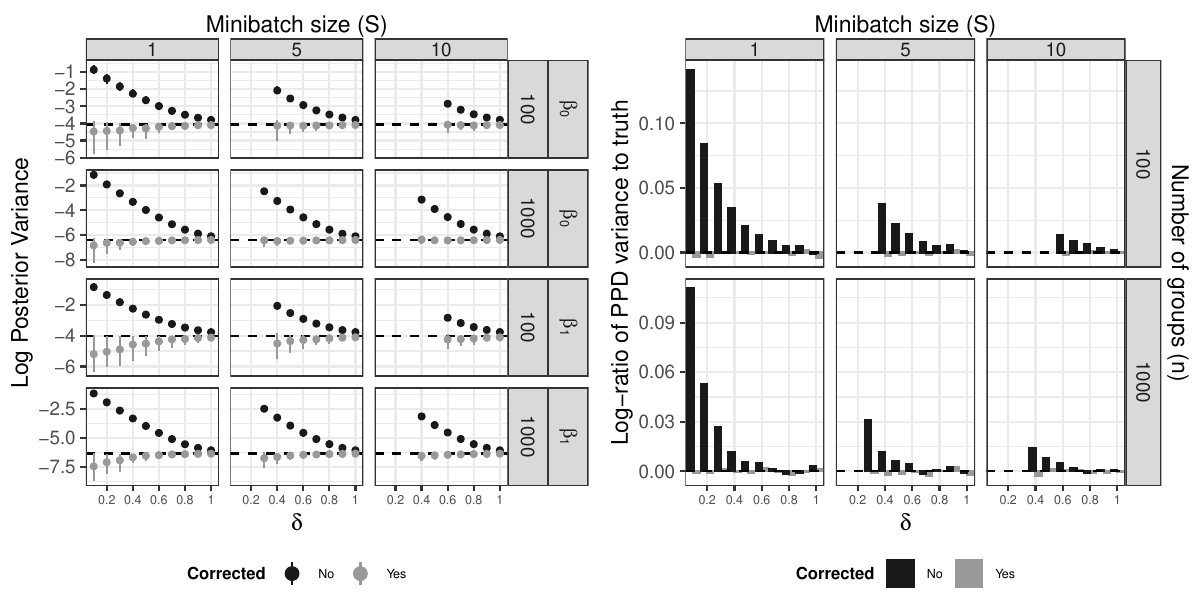}
\caption{(Left) Comparison of SGLD posterior variance estimates from the experiment in Section \ref{sims:fixed} against the truth (dashed black line), before and after correction. 
Points and intervals correspond to mean and 95\% quantile-based intervals across replicate datasets.
(Right) Log ratio of SGLD-estimated PPD variance to the true PPD variance, before and after correction.}
\label{fig:sim1_variance}
\end{center}
\end{figure}

Results from this simulation are presented in Figure \ref{fig:sim1_variance}, with more details provided in Figures 1--3 of the online supplementary materials \citep{oursupplement}. 
The left side of Figure \ref{fig:sim1_variance} presents the log-posterior variance estimates, averaged across 100 simulated data sets with 95\% quantile-based intervals. The results indicate that the uncorrected SGLD algorithm consistently overestimated the posterior variance, though the upward bias reduces as $\delta \rightarrow 1$. Higher $\delta$ yields better accuracy, but it implies smaller step size and a computational overload. The corrected SGLD improved the accuracy: quantile intervals always included the true log-variance, and the estimation quality improved with larger $n$, $S$, and $\delta$. These results emphasize the importance of choosing $\delta$ that balances the desired accuracy with computational budget. 
The right side of Figure \ref{fig:sim1_variance} assesses the predictive accuracy using the algorithmic estimates, through the log ratio of estimated-to-true PPD variances. Log ratio values near zero indicate good performance.  
The corrected SGLD maintains proper uncertainty quantification in the PPD in every simulation setting. The uncorrected SGLD algorithm generally performs poorly, with acceptable performance only when $\delta \rightarrow 1$. As a reminder, SGLD requires longer runtime with larger $\delta$ and when the sample size grows values of $\delta$ close to 1 are not computationally feasible. 
Thus, the corrected SGLD is a viable alternative to uncorrected SGLD when the sample size becomes large, because it can yield proper PPD estimation regardless of $\delta$, and therefore is computationally efficient.

While this experiment described fitting LMMs with known variance components, in real world settings, both the dispersion parameter of the Gaussian likelihood and the covariance parameter of subject-specific coefficient will also need to be estimated.
The performance of our method in a more realistic LMM is given in Section 6.3 of the online supplementary materials \cite{oursupplement}.
Since the main aim of our method is to address inference in large data non-Gaussian GLMMs, we move on to describing simulated studies for binary and count data outcomes.

\subsection{Logistic GLMM for Binary Outcome Data}
\label{sims:unknown_bin}

In this Section, we assess the accuracy of the proposed method in fitting Bayesian logistic GLMM to binary outcome data.
Recall that this model is obtained from the general framework \eqref{e:exp_family} by setting $a(\phi) = 1$, $b(\theta_{it}) = \log(1 + \exp\{\theta_{it}\})$, $c(Y_{it}, \phi) = 0$, $\phi = 1$, and $\theta_{it} = \log(\pi_{it} / (1 - \pi_{it}))$. The mean parameter is given by $\mu_{it} = \pi_{it} = \exp\{\theta_{it}\} / (1 + \exp\{\theta_{it}\})$. This yields the model, $Y_{it}|\boldsymbol{\gamma}_i,\boldsymbol{\Omega} \stackrel{{ind}}{\sim} \text{Bernoulli}(\pi_{it})$, $\text{logit}(\pi_{it}) = \mathbf{x}_{it}^\top\boldsymbol{\beta} + \mathbf{z}_{it}^\top\boldsymbol{\gamma}_i$, where $\text{logit}(x) = \log(x/(1-x))$, and the population parameters are given by, $\boldsymbol{\Omega} = \{\boldsymbol{\beta},\boldsymbol{\Sigma}\}$. We assumed that $\mathbf{x}_{it} = \left(1,x_{it}\right)^\top$ and $x_{it} \stackrel{{iid}}{\sim} \mathcal N(0,1)$. True parameters values are the same as in Section \ref{sims:fixed}. 
We generated 100 datasets using these settings, for two different group sizes  $n \in \{10^4,10^5\}$ with $n_i = 10$.

For each sampling algorithm, we used re-parameterized between-group covariance parameters as follows:
$\delta_{1} = \log (\sigma_{1})$, $\delta_{2} = \log (\sigma_{2})$, and $\delta_{\rho} = 2{\rm atanh}(\rho)$. 
Independent half-$t$ priors with 3 degrees of freedom were placed on each standard deviation. A uniform prior was placed on $\rho$, a choice which can be generalized to handle multivariate correlation matrices \cite{lewandowski2009generating}.
An independent $N(0, 10^2)$ prior was used for the intercept and slope. 
Gradients for the conditional log-likelihood and prior are reported in Section 6 of the online supplementary \citep{oursupplement}. Within the algorithm, we used $R=1,000$ P\'{o}lya-Gamma data augmentation Gibbs samples targeting $p(\boldsymbol{\gamma}_i | \mathbf{Y}_i, \boldsymbol{\Omega})$ for each $i$ at each iteration; this procedure, along with more details about the experiment, is detailed in the online supplementary materials \citep{oursupplement}.

The two main tuning parameters of SGLD algorithm are minibatch size $S$ and hyperparameter $\delta>0$ which determines the step size as $\epsilon = S/n^{1+\delta}$. 
We considered three choices of $S \in \{1,5,10\}$, based on practical needs (i.e., $S$ cannot be too big for computational gains).
For each $S$, we chose a $\delta$ using the heuristic rule: $\delta = (\delta_{\text{min}} + \delta_{\text{max}}) / 2$, where $\delta_{\text{min}} = \min\{\delta\in \{0.1,0.2,\ldots,1\} : \epsilon < n^{-1}\}$ and $\delta_{\text{max}}$ = 1 (recall that $\epsilon = n^{-(1 + \delta)}S$). The intuition behind this definition is that an ideal value of $\delta$ will be in a range with proper posterior variance estimation while maximizing computational efficiency.

\begin{figure}[t]
\begin{center}
\includegraphics[width=\textwidth]{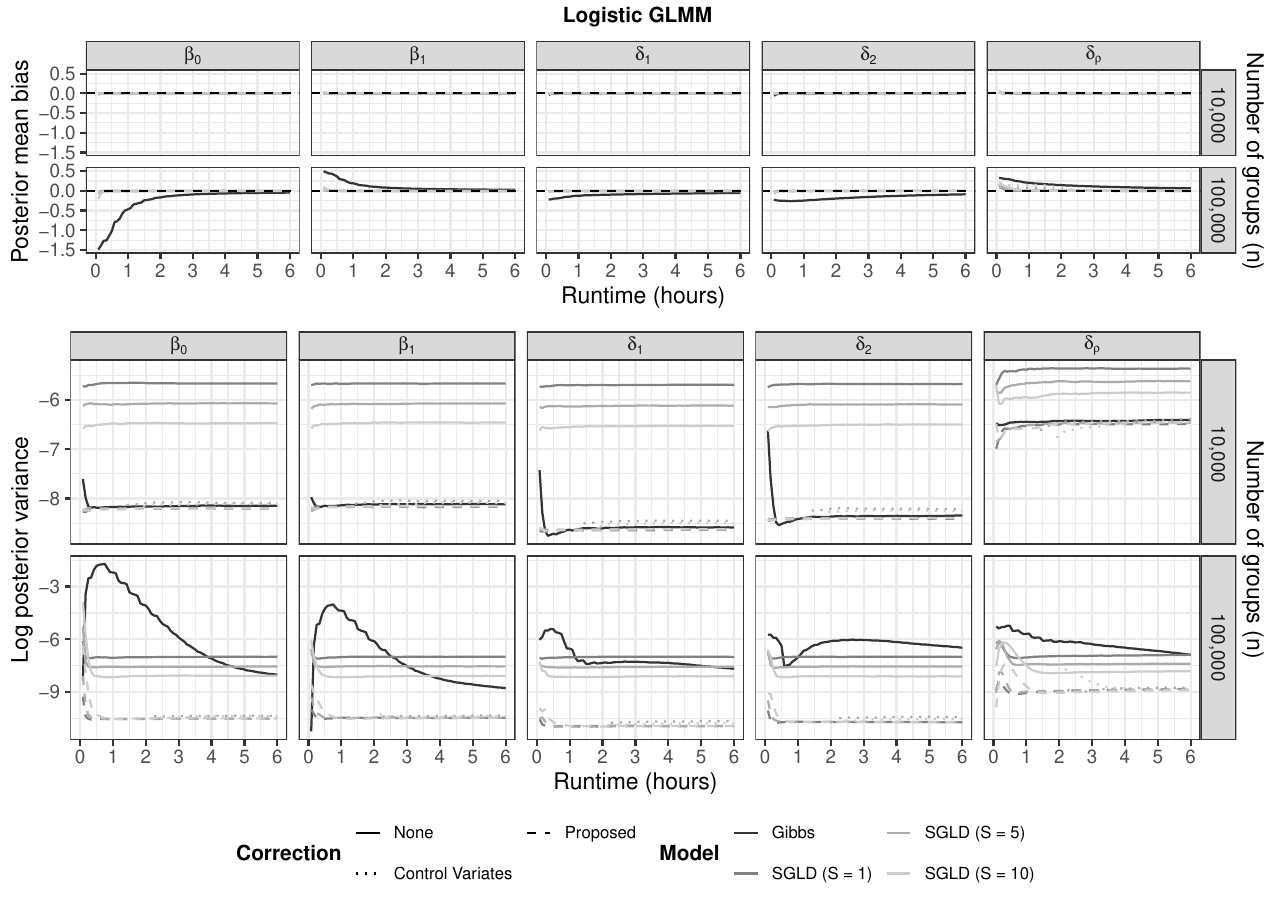}
\caption{Bias of posterior mean estimate and (log) posterior variance estimate for each algorithm as a function of runtime for the experiment in Section \ref{sims:unknown_bin}.
Each summary statistic is averaged over 100 replicate datasets.} 
\label{fig:sim3_variance}
\end{center}
\end{figure}

\subsubsection*{Performance Evaluation and Comparison}

Performance of SGLD was assessed by pre-specifying a runtime window of 6 hours and evaluating online posterior mean/variance estimates as functions of runtime, each run parallel across 100 generated datasets.
Posterior mean and variance estimates were aggregated to produce Figure.
The accuracy of posterior mean estimate was assessed against the true parameters used for simulating the data.
The accuracy of posterior variance estimate cannot be so easily assessed, however, posing a benchmark problem.
For both $n=10^4,10^5$, we considered comparison against a Polya-Gamma data augmentation-based Gibbs sampler with no data subsampling.
The comparison served a dual purpose, one of providing a benchmark posterior variance estimate (provided it converges in a reasonable time) and demonstrating possible computational gains of SGLD (by faster convergence in real runtime). 
Furthermore, for $n=10^4$ only, we also considered sampling from a No U-Turn Sampler (NUTS) implemented in STAN software \cite{stan}, as a further benchmark. 
The default warm-up and initialization settings were used, producing 1,000 samples in 4 parallel chains, for each of the generated datasets. 
NUTS was not considered for $n=10^5$ setting, because of a longer required runtime and severe memory constraints.

In addition to the already described Gibbs and NUTS benchmarks, we considered the SGLD-CV algorithm proposed by \cite{baker2019control}, a variant of SGLD algorithm using control variates, as a competitor method.
Making direct comparisons to the previous authors' work in the GLMM setting is not straightforward, as their work did not consider the GLMM setting. We thus implemented a custom control variates method based on their work. First, in the initially required SGD phase, we used $R=5,000$ samples so that this competitor algorithm has a better chance of accurately locating the posterior mode. Second, we used two independent sets of MCMC samples $\boldsymbol{\gamma}_i$, targeting $p(\boldsymbol{\gamma}_i|\mathbf{Y}_i,\OO_k)$ at each iteration $k$ and $p(\boldsymbol{\gamma}_i|\mathbf{Y}_i,\hat{\OO}^*)$ for the SGD esitmate of the MLE $\hat{\OO}^*$. Finally, we fixed the step size in the sampling phase to be $O(1/n^{1+\delta})$ with $\delta = 0.4$, as larger step sizes tended to lead to divergence. 

\subsubsection*{Variance Estimation and Coverage Results}

Figure \ref{fig:sim3_variance} summarizes the posterior summary estimates obtained from various algorithms as a function of runtime.
For a fair comparison of the accuracy obtained by our proposal with other methods, the proposed correction was periodically computed in an online manner, using samples obtained prior to each 5-minute mark.
The key take-aways are:
\begin{enumerate}
    \item for $n=10^4$, the corrected SGLD algorithm yielded posterior variance estimates indistinguishable from those of the Gibbs sampler, while the uncorrected SGLD algorithm overestimated the variance;
    \item the performance of the corrected SGLD algorithm was not impacted by the increase in sample size, as opposed to the Gibbs sampler or NUTS, which severly suffered from a computational overload;
    \item SGLD-CV results were mostly comparable to our method, though posterior variance estimates were larger for all parameters at termination of the algorithm, and convergence for $\delta_\rho$ was slower for $n=10^5$;
    \item the above findings were robust to the choice of $S$ or $n$.
\end{enumerate}
The increased accuracy of variance estimation also translated to good empirical coverage properties of marginal quantile-based credible intervals obtained from the corrected samples. Table \ref{tab:cover1} shows these results.
For a further confirmation of the accuracy of posterior variances in logistic GLMMs, we compared the variance estimates against NUTS run for 6 hours. The comparison is shown in Figure \ref{fig:stan_variance} and demonstrate again the accuracy of our corrected posterior variance estimates.
Additional details, including the precise expression for the model gradients and runtimes for the post-sampling correction algorithm, are reported in Section 6.2 of the online supplementary materials \cite{oursupplement}.

Our implementation of SGLD-CV yielded posterior variance estimates mostly comparable to the proposed method.
They were consistently larger than our corrected variance estimates across all parameters and all settings, with the difference between them larger when $n = 10^4$.
The variance inflation is consistent with the fact that a naive implementation of SGLD-CV adapted to GLMMs, unlike our proposed correction, does not correct for the additional variability of the stochastic gradient induced by using MCMC samples to approximate the gradient of the marginal log-likelihood. 
Moreover, the slower convergence of posterior variance estimate for $\delta_\rho$ for SGLD-CV with $n=10^5$ suggests that the advantage gained by the use of control variates may depend on model parameterization.
A more involved control variates implementation that that further reduces the stochastic gradient variance, through using Fisher's identity and then coupling the \emph{inner loop} Markov chains, may outperform both methods. Such an implementation is promising but goes beyond the scope of this work.

\subsection{Count Data Modeling with Poisson Likelihood}
\label{sims:unknown_count}

In this Section, we assess the accuracy of the proposed method in fitting Bayesian Poisson GLMM to count outcome data.
Recall that this model is obtained from the general framework \eqref{e:exp_family} by setting
$a(\phi) = 1$, $b(\theta_{it}) = \exp\{\theta_{it}\}$, $c(Y_{it}, \phi) = -\log(Y_{it}!)$, $\phi = 1$, and $\theta_{it} = \log(\lambda_{it})$. The mean parameter is given by $\mu_{it} = \lambda_{it} = \exp\{\theta_{it}\}$. This yields the model, $Y_{it}|\boldsymbol{\gamma}_i,\boldsymbol{\Omega} \stackrel{\text{ind}}{\sim} \text{Poisson}(\lambda_{it})$, where $\log(\lambda_{it}) = \mathbf{x}_{it}^\top\boldsymbol{\beta} + \mathbf{z}_{it}^\top\boldsymbol{\gamma}_i$, and $\boldsymbol{\Omega} = \{\boldsymbol{\beta},\boldsymbol{\Sigma}\}$. In this simulation, the diagonal values of $\boldsymbol{\Sigma}$ were changed to 0.4 to generate more realistic Poisson random variables; otherwise, the details of each simulated dataset, prior specifications, and unconstrained re-parameterizations of $\boldsymbol{\Sigma}$, were all set to be the same as the logistic GLMM experiment described in Section \ref{sims:unknown_bin}.
Unlike in the case of logistic GLMM, there is no easy data augmentation procedure that can be used to produce samples from the conditional distribution of subject-specific parameters. 
For our SGLD algorithm and the compared Gibbs sampler, we used a Metropolis-Hastings sampler for estimating the stochastic gradient.
Proposal values were sampled from $\boldsymbol{\gamma}_{ir}^* \sim \mathcal{N}\left(\boldsymbol{\gamma}_{ir},\boldsymbol{\Delta}\right)$, where $\boldsymbol{\Delta}$ is a tuning parameter adaptively chosen \cite{Roberts01012009}. 
The algorithm was implemented with $R=$5,000 after a burn-in period and the acceptance rate was tuned to be approximately 40\%. 

\begin{figure}[t]
\begin{center}
\includegraphics[width=\textwidth]{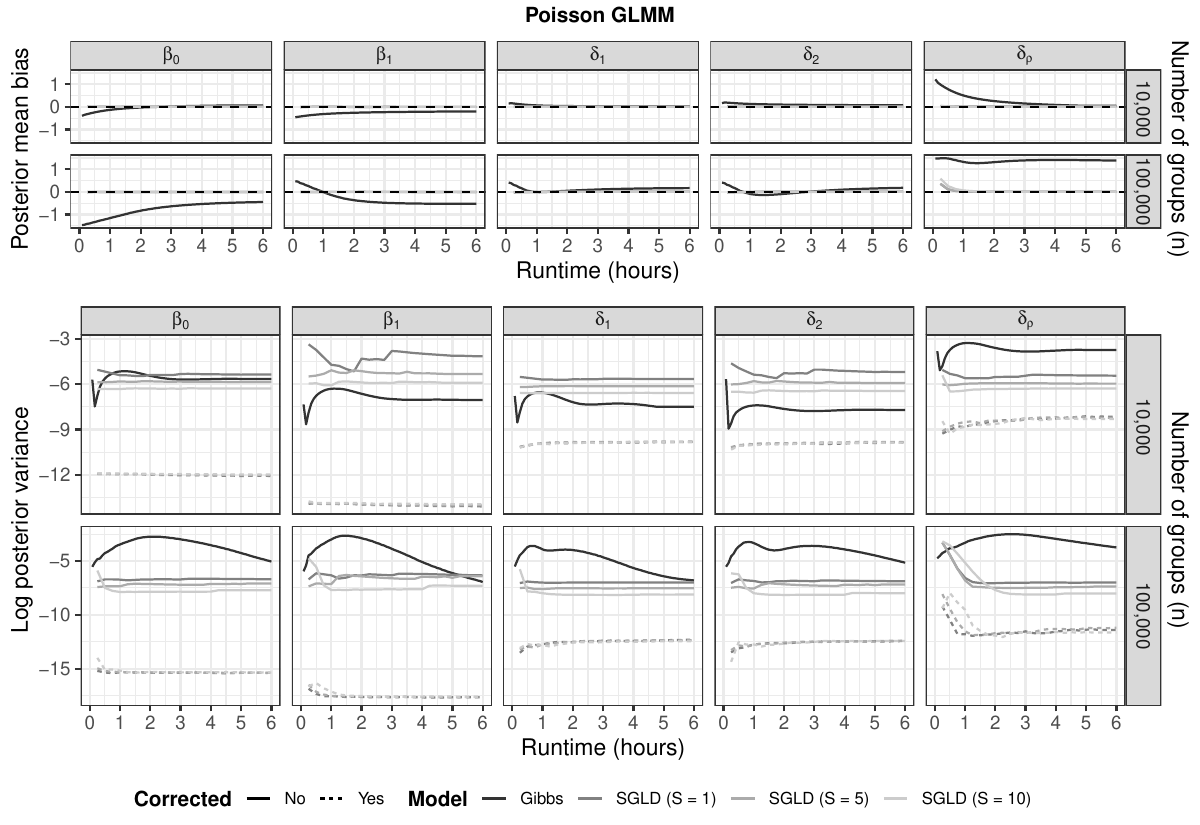}
\caption{Bias of posterior mean estimate and (log) posterior variance estimate for each algorithm as a function of runtime for the experiment in Section \ref{sims:unknown_count}.
Each summary statistic is averaged over 100 replicate datasets.}
\label{fig:sim4_variance}
\end{center}
\end{figure}

This experiment posed many computational problems uniformly across all methods: the use of Metropolis-Hastings sampler within the Gibbs iteration visibly slowed down its convergence, and NUTS sampler for $n=10^4$ took significantly more runtime than the allotted 6 hours to produce the first 1,000 samples on average. For the latter sampler, we stopped sampling at the 1,000 sample mark if runtime exceeded 6 hours.
Additional results, including the Metropolis-within-Gibbs sampler implementation and runtimes for the post-sampling correction algorithm, are reported in Section 6.3 of the online supplementary materials \cite{oursupplement}.

\subsubsection*{Variance Estimation and Coverage Results}

\begin{table}[t]
\centering
\label{tab:cover1}
\caption{Empirical coverage (\%; standard error) of 95\% marginal quantile-based credible intervals.}
\small
\setlength{\tabcolsep}{3pt}
    \begin{tabular}{|l|c|c|c|c|c|c|c|c|c|c|} 
     \hline
     & \multicolumn{5}{|c|}{Logistic} & \multicolumn{5}{|c|}{Poisson}   \\
     \hline
      & $\beta_0$ & $\beta$ & $\delta_1$ & $\delta_2$ & $\delta_\rho$ & $\beta_0$ & $\beta$ & $\delta_1$ & $\delta_2$ & $\delta_\rho$\\ [0.5ex] 
     \hline\hline
     $n=10^4,\;S=1$ & $ 97~(2)$ & $98~(1)$ & $88~(3)$ & $92~(3)$ & $94~(2)$ &
     $ 87~(3)$ & $82~(4)$ & $82~(4)$ & $84~(4)$ & $47~(5)$
     \\ 
     $n=10^4,\;S=5$ & $97~(2)$ & $99~(1)$ & $92~(3)$ & $92~(3)$ & $99~(1)$ &
    $85~(4)$ & $85~(4)$ & $87~(3)$ & $89~(3)$ & $53~(5)$ 
     \\
     $n=10^4,\;S=10$ & $100~(0)$ & $98~(1)$ & $92~(3)$ & $92~(3)$ & $98~(1)$ &
     $86~(4)$ & $83~(4)$ & $94~(2)$ & $95~(2)$ & $50~(5)$\\
     \hline\hline
      $n=10^5,\;S=1$ & $98~(1)$ & $99~(1)$ & $96~(2)$ & $88~(3)$ & $94~(2)$ &
      $72~(4)$ & $48~(5)$ & $82~(4)$ & $80~(4)$ & $21~(4)$\\
     $n=10^5,\;S=5$ & $98~(1)$ & $97~(1)$ & $95~(2)$ & $97~(2)$ & $95~(2)$ &
     $76~(4)$ & $50~(5)$ & $90~(3)$ & $78~(4)$ & $29~(5)$\\
     $n=10^5,\;S=10$ & $99~(1)$ & $99~(1)$ & $96~(2)$ & $97~(2)$ & $94~(2)$ &
     $85~(4)$ & $62~(5)$ & $93~(3)$ & $87~(3)$ & $41~(5)$\\
     \hline
    \end{tabular}
\end{table}

\begin{figure}[t]
    \centering
    \includegraphics[width=0.9\textwidth]{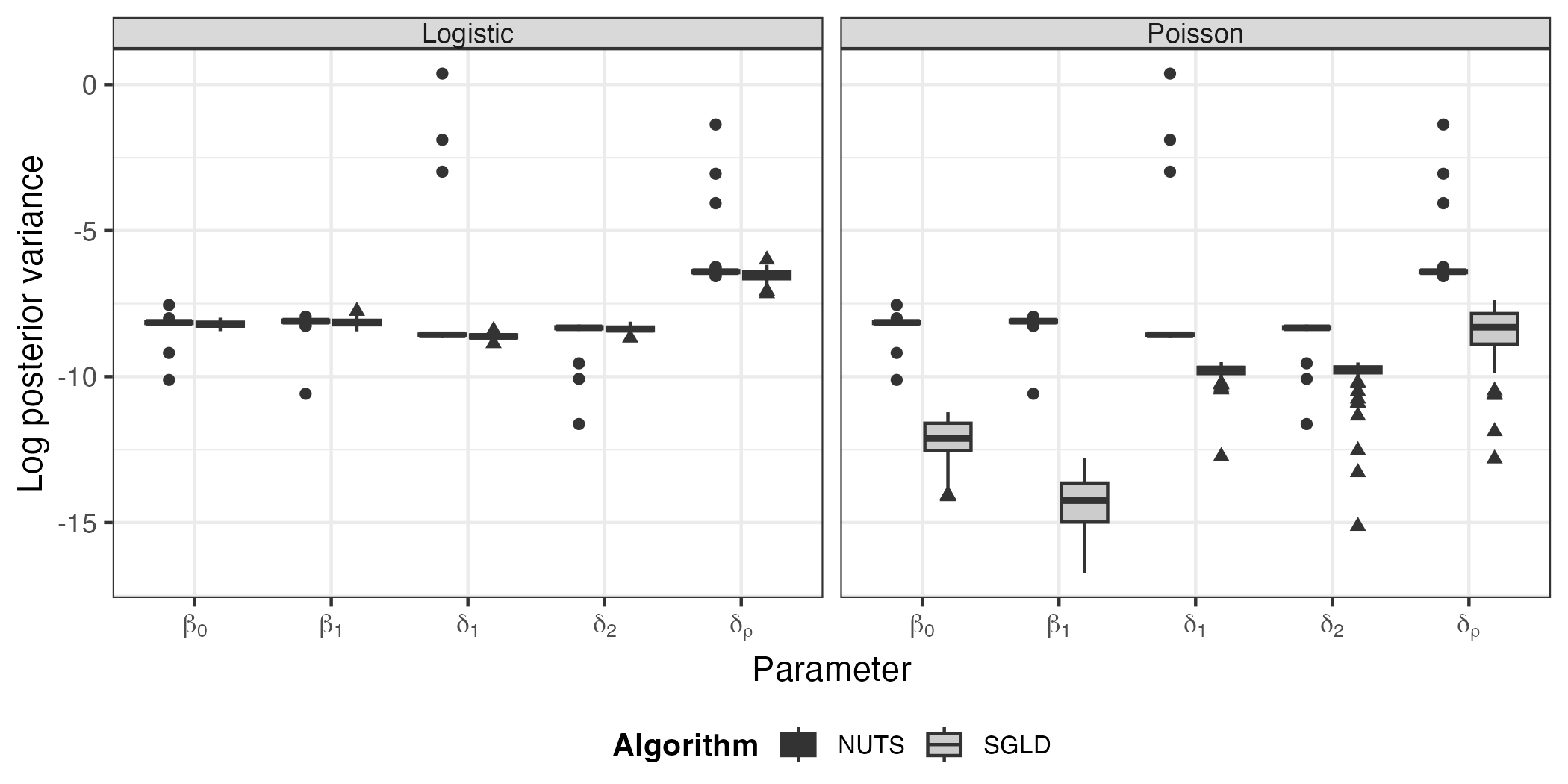}
    \caption{Comparison of corrected SGLD posterior variance estimates against NUTS for logistic and Poisson GLMMs across 100 simulated datasets.}
    \label{fig:stan_variance}
\end{figure}

Figure \ref{fig:sim4_variance} summarizes the posterior summary statistics obtained from various algorithms as a function of runtime.
Table \ref{tab:cover1} shows the empirical coverage results, and Figure \ref{fig:stan_variance} compares the corrected SGLD variance estimates against that obtained from NUTS run for 6 hours.
When compared against results from the logistic GLMM, there was a noticeable deterioration in both the empirical coverage of corrected sample-based credible intervals, and the accuracy of corrected SGLD-based posterior variance estimates.
Comparison with NUTS suggests that the correction method consistently appears to \emph{over-correct}, by estimating a marginal posterior variance much lower than NUTS (if we take the latter as the benchmark).
This underestimation issue was more severe for the two fixed effect parameters and the correlation parameter $\delta_\rho$, for which coverage was also generally poorer than for $\delta_1$ and $\delta_2$.
This over-correction does not appear to be sufficiently explained by a possible under-estimation of the variance of the stochastic gradient (cf. the discussion at the end of Section \ref{sec:theory}), for this should also affect the logistic regression simulation, in which case our method showed excellent accuracy.
There are two possible issues here. First, the log-link Poisson model with normal random effects does not appear to be a reliable generative model for simulations, since the count maxima of groups relative to the medians were often unrealistically large. Second, the super-exponentially decaying tail behavior of the posteriors for $\boldsymbol{\beta}$, due to the commonly used exponential link parameterization, which violates assumption \ref{a:smooth} in Section \ref{sec:theory}.
This issue has prompted other authors to either analyze a different link function \cite{prado2026metropolishastingsscalablesubsampling} or a theoretical analysis favoring Gibbs samplers that do not use gradient information \cite{chak2025complexity}.

Although our correction does not perfectly estimate the posterior variance for $\boldsymbol{\beta}$ as was the case for the logistic GLMM, it substantially mitigates the extent to which SGLD overestimates the posterior variance in light of the STAN benchmark.
These results also suggest that the different posterior behavior suggested by Chak and Zanella \cite{chak2025complexity} is not a theoretical artifact, and that further studies may be needed in the future for reliable Bayesian inference with big count data.

\section{Real World Data Analysis}
\label{sec:da}

In patients with ophthalmic disorders, psychiatric risk factors play an important role in morbidity and mortality. Understanding how patient characteristics impact the probability of a patient having distress is a critical task that will facilitate proper and early psychiatric screening and result in prompt intervention to mitigate its impact. In this analysis of real world clinical data, we analyzed data from the Duke Ophthalmic Registry, a database that consists of adults at least 18 years of age who were evaluated at the Duke Eye Center or its satellite clinics from 2012 to 2021. The goal of the analysis was to identify patient characteristics associated with a diagnosis of psychiatric distress upon each clinic encounter. The Duke University Institutional Review Board approved this study with a waiver of informed consent due to the retrospective nature of this work. All methods adhered to the tenets of the Declaration of Helsinki for research involving human subjects and were conducted in accordance with regulations of the Health Insurance Portability and Accountability Act.

We defined ${Y}_{it} \in \{0,1\}$ as an indicator of psychiatric distress for patient $i$ ($i\in[n]$) and encounter $t$ ($t \in [n_i]$). Distress was defined using an electronic health records phenotype that has been detailed previously \citep{berchuck2022framework}. Each encounter was observed at follow-up time $\tau_{it}$, where $\tau_{i1}=0$, such that $\tau_{it}$ for $t > 1$ indicates the number of years from baseline encounter. The encounters for each patient were collected in the vector $\mathbf{Y}_i$. We also defined an indicator $w_i = 1\left(\sum_{t=1}^{n_i} Y_{it} = 0\right)$ to indicate whether all observations for a patient were zero. The observed data is given by $(\mathbf{Y}_i, w_i)$ with the joint distribution $p(\mathbf{Y}_i, w_i) = p(w_i)p(\mathbf{Y}_i | w_i)$. The outcome was modeled as $Y_{it} \stackrel{{ind}}{\sim} \text{Bernoulli}(\pi_{it})$, with $\text{logit}(\pi_{it}) = \mathbf{x}_{it}^\top \boldsymbol{\beta} + \mathbf{z}_{it}^\top \boldsymbol{\gamma}_i$ and the missing indicator was modeled as $w_{i} \stackrel{{ind}}{\sim} \text{Bernoulli}(p_{i})$, where $\text{logit}(p_{i}) = \alpha_0 + \mathbf{x}_{i}^\top \boldsymbol{\alpha}_{-0}$ and $\boldsymbol{\alpha} = (\alpha_0,\boldsymbol{\alpha}_{-0}^\top)^\top$. The covariates were defined as follows, $\mathbf{x}_{it} = (1,\tau_{it},\mathbf{x}_i^\top)$, where $\mathbf{x}_i$ contains patient-level covariates (e.g., baseline age). The vector $\mathbf{z}_{it} = (1,\tau_{it})^\top$ inducing a subject-specific intercept and slope for follow-up time, such that $q = 2$.

The set of population parameters is $\boldsymbol{\Omega} = (\boldsymbol{\alpha},\boldsymbol{\beta},\boldsymbol{\Sigma})$ and the posterior distribution is,
$p(\boldsymbol{\Omega} | \mathbf{Y}, \mathbf{w}) \propto p(\boldsymbol{\Omega})\prod_{i=1}^n p({w}_i | \boldsymbol{\alpha})\int p(\mathbf{Y}_i | \boldsymbol{\gamma}_i , {w}_i, \boldsymbol{\beta})p(\boldsymbol{\gamma}_i | \boldsymbol{\Sigma}) d\boldsymbol{\gamma}_i$,
with corresponding gradient,
$$\nabla \log p(\boldsymbol{\Omega} | \mathbf{Y},\mathbf{w}) = \nabla \log p(\boldsymbol{\Omega}) + \sum_{i=1}^n \nabla \log p(w_i|\boldsymbol{\alpha}) + \mathbb{E}_{\boldsymbol{\gamma}_i|\mathbf{Y}_i,w_i,\boldsymbol{\Omega}}[\nabla \log p(\mathbf{Y}_i,\boldsymbol{\gamma}_i | w_i,\boldsymbol{\Omega})],$$
where the expectation can be approximated using \eqref{e:hatgi}. The gradients inside of the expectation are $\nabla_{\boldsymbol{\beta}} \log p(\mathbf{Y}_i | \boldsymbol{\gamma}_{ir},\boldsymbol{\beta}, w_i) = 1(w_i = 0) \nabla_{\boldsymbol{\beta}} \log p(\mathbf{Y}_i | \boldsymbol{\gamma}_{ir},\boldsymbol{\beta})$, $\nabla_{\boldsymbol{\Sigma}} \log p(\boldsymbol{\gamma}_{ir} | \boldsymbol{\Sigma})$ is defined as before and $\nabla_{\boldsymbol{\alpha}} \log p(w_i | \boldsymbol{\alpha}) = (w_i - p_i)\mathbf{x}_i^\top$. Samples of $\boldsymbol{\gamma}_{ir}$ can be obtained from $p(\boldsymbol{\gamma}_i | \mathbf{Y}_i,w_i,\boldsymbol{\Omega}) = w_i p(\boldsymbol{\gamma}_i | \boldsymbol{\Sigma}) + (1 - w_i) p(\boldsymbol{\gamma}_i | \mathbf{Y}_i,\boldsymbol{\Omega})$. It is clear that when $w_i = 1$ the posterior is equal to the prior and when $w_i = 0$ the posterior is the same as the Bernoulli GLMM.

The sample included 40,326 patients ($n$), of which 15\% had at least one distress indicator. On average patients had 9 encounters, were 60 years of age, 41\% male, and 68\% white. Full demographic details are given in Table 3 of the online supplementary materials \citep{oursupplement}. In our model, the reference categories for categorical variables were White (race) and commercial (insurance). The continuous variables age, income, and education were standardized.

\begin{figure}[t!]
\begin{center}
\includegraphics[width=0.99\textwidth]{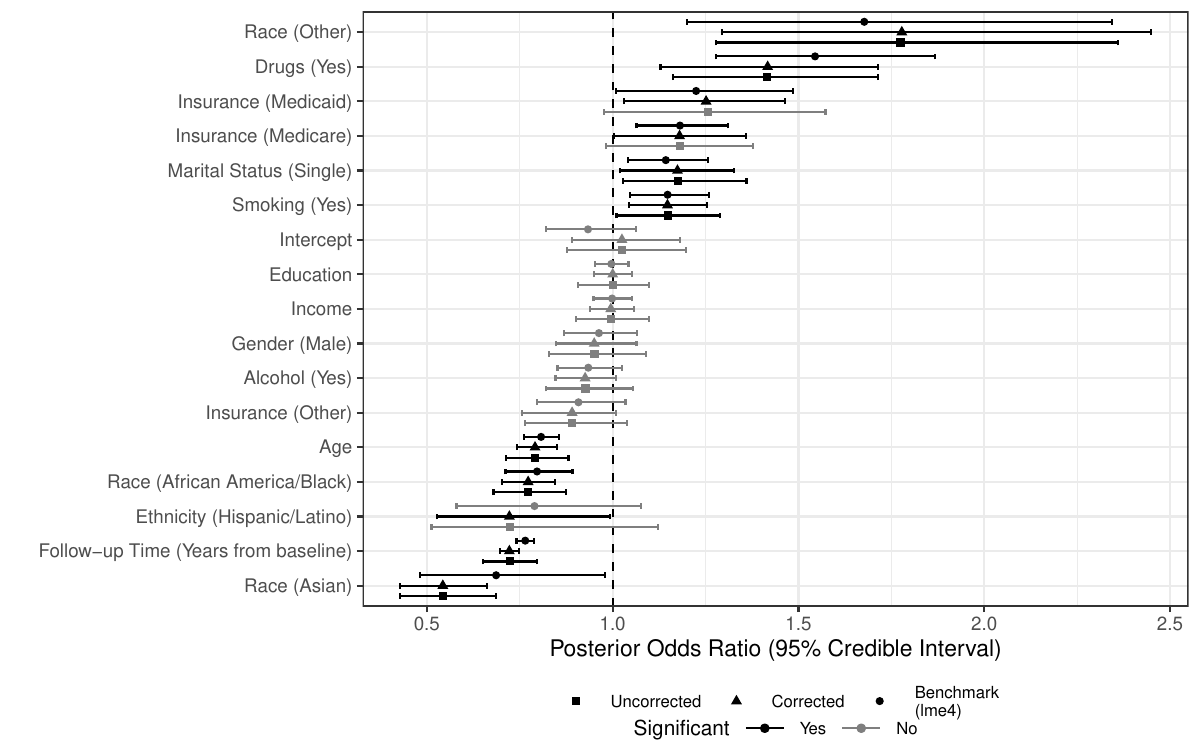}
\caption{Posterior odds ratios and 95\% credible intervals for $\boldsymbol{\beta}$. Summaries are presented for the uncorrected and corrected SGLD algorithm, along with
95\% confidence intervals based on maximum likelihood estimates from fully observed data using the \texttt{lme4} package \citep{bates2015fitting}.
The parameters are presented in decreasing order and color coded based on whether the credible interval contained zero. \label{fig:realdata_betas}}
\end{center}
\end{figure}

We performed inference using both the corrected and uncorrected SGLD algorithms using $S=1$. In Figure \ref{fig:realdata_betas}, we presented odds ratios and 95\% credible intervals for $\boldsymbol{\beta}$ for both SGLD algorithms. Estimates are color-coded to indicate whether the 95\% credible interval included one, which corresponds to a one-sided Bayesian p-value. Of the variables, other race, drugs use, Medicaid insurance, Medicare insurance, single marital status, and smoking were all significantly associated with an increased probability of being distressed, assuming the presence of any distress (i.e., $w_i = 1$). An increase in baseline age, Black/African American race, increased follow-up time, Hispanic/Latino ethnicity, and Asian race were all associated with a decreased probability of being distressed, 
though the race Asian and Hispanic/Latino ethnicity coefficient estimates may not have enough clinical generalizability due to low counts (103 and 172 patients with non-missing distress outcomes, respectively).
Examination of this figure illustrates the importance of the covariance correction for obtaining proper inference as the uncorrected SGLD algorithm yielded insignificant p-values for Medicaid insurance, Medicare insurance, and Hispanic/Latino ethnicity. 
The posterior estimates of $\boldsymbol{\alpha}$ and $\boldsymbol{\Sigma}$ are given in Figure 6 and Table 4 of the online supplementary materials \citep{oursupplement}.

To further assess the validity of the different conclusion based on correcting the posterior samples, we have also compared inference for $\boldsymbol\beta$ only with logistic mixed model maximum likelihood estimates on the fully observed data using the \texttt{lme4} package \citep{bates2015fitting}. 
These may be taken as an approximation of the Bayes estimates under diffuse priors on all model parameters, though it uses Laplace approximation for the marginal likelihood and is known to be inaccurate in many cases \cite{hall2020fast}.
Nevertheless, comparison of \texttt{lme4} confidence intervals and the SGLD posterior credible intervals shows that the two methods yield comparable interval estimates for all variables except Asian race coefficient. Furthermore, \texttt{lme4}-based confidence intervals yielded significant $p$-values for the Medicaid and Medicare insurance variables, agreeing with those for the corrected SGLD algorithm and supporting our claim that posterior correction materially changes scientific conclusions.

\section{Discussion}

In this paper, we introduced an algorithm for performing accurate and scalable Bayesian inference for the GLMM in the large-sample regime. To the best of our knowledge, our approach is the first to adapt SGMCMC methods to this setting, overcoming limitations that arise from a naive application of such algorithms to dependent data. The main contributions of our work include: (i) a Monte Carlo approximation of the gradient of the marginal log-likelihood using Fisher’s identity, enabling SGMCMC updates in the presence of intractable integrals; and (ii) an asymptotic analysis of the variance structure of the resulting algorithm, allowing us to derive a post-hoc correction that restores the correct posterior covariance. 
We demonstrated the accuracy and scalability of our algorithm for Bayesian inference in GLMMs through a series of simulation studies in Section \ref{sims}, with models for Gaussian, logistic (binary outcome), and Poisson (count outcome) data. In particular, Sections~\ref{sims:unknown_bin}--\ref{sims:unknown_count} demonstrated how SGLD is a promising approach for GLMM inference that maintains computational efficiency as $n$ increases, even in regimes where MCMC becomes prohibitively expensive.

As discussed in the introduction, fixed step size SGLD leads to inflated posterior variance \citep{brosse2018promises}, and a variety of correction strategies have been developed to address this, including preconditioning \citep{stephan2017stochastic}, Metropolis-Hastings adjustments \citep{gelman1995bayesian}, and control variate methods \citep{baker2019control}. However, these approaches are typically designed for independent data settings and can be prohibitively expensive or difficult to extend to models with intractable marginal likelihoods, such as GLMMs. In our simulations, we compared our method directly to the SGLD-CV method, a scalable method with good theoretical properties. While SGLD-CV effectively reduces variance from minibatching, it fails to account for the additional variability introduced by approximating the marginal likelihood. 
This omission results in miscalibrated posterior uncertainty in GLMMs. In contrast, our correction accounts for both sources of variability and achieves accurate uncertainty quantification, as demonstrated in both posterior and predictive comparisons across a range of GLMM settings. 
These results highlight the importance of tailoring variance correction to the specific structure and inferential challenges of dependent-data models.
In Section~\ref{sec:da}, we applied our algorithm to a real-world electronic health records dataset to demonstrate its utility for statistical inference in complex, large-scale settings. As shown in Figure~\ref{fig:realdata_betas}, the covariance correction meaningfully impacted inference: 95\% credible intervals changed interpretation for key covariates such as insurance status and ethnicity, underscoring the importance of properly accounting for variance inflation. 

The term ``mixed models'' can refer to a broader suite of models than what have been described here. 
Dataset can have multiple levels of nesting, as when longitudinally correlated data are available for each patient hospitalized in different medical centers.
In principle, computational methods described in this paper admit an extension to multiple nesting settings, but some details may need to be worked out for the variance correction formula depending on the form of the model accommodating multilevel random effects.
A more involved class of models are the so-called ``crossed effects'' models, where the data structure does not admit a single global nesting variable.
These models are known to pose significantly more computational challenges than mixed models considered here, and have sparked a flurry of recent research in both likelihood-based \cite{ghosh2022backfitting,bellio2025consistent} and Bayesian inference \cite{papaspiliopoulos2020scalable,papaspiliopoulos2023scalable}.
Crossed-effects extension of the methods described in this work will require a substantial modification, due to the lack of a simple data subsampling scheme that yields an even theoretically unbiased estimator, and will be pursued elsewhere.

There are other possible improvements and limitations to be noted.
One of the pronounced limitations of this work was that the accuracy of the proposed method is much worse for exponential link Poisson GLMM than for logistic GLMM. 
A possible remedy is to change the link function, similarly as done by Prado et al. \cite{prado2026metropolishastingsscalablesubsampling} or investigated in more detail by Hamura \cite{hamura2025exact}.
This will unfortunately imply that using the traditionally favored exponential link is foreclosed, with less interpretable coefficients.
Future research may need to address adaptive step-size discretizations that can overcome the smoothness assumption \ref{a:smooth}.
The second immediate limitation is the need to choose $R$, the number of iterations, and possibly more tuning parameters, of an inner loop MCMC algorithm for computing the stochastic gradient estimator.
To eliminate tuning needs and possibly also gain algorithmic efficiency, the authors have begun an investigation into alternative means of stochastic gradient computation needed for GLMMs.
The use of expectation propagation in the context of GLMMS was first investigated by Hall et al. \cite{hall2020fast} for probit models and appears highly promising.
Other promising approaches include a combination of methods introduced in this work with quasi-Monte Carlo, rather than Monte Carlo, samples \cite{sobol1990quasi}.

At a theoretical level, future research directions include obtaining quantitative concentration bounds for the uncorrected posterior distribution, giving indications on how to choose the parameter $\delta$ depending on $n$. It will be also important to consider extensions of our algorithm to modern machine learning settings, to incorporate momentum in stochastic gradient updates, provide guarantees for high-dimensional models, and infer deep neural networks. 
At an application level, future research directions include generalizations to non-trivial distributional assumptions and covariance specifications for the subject-specific parameters, flexible prior specifications for the population regression parameter to allow for sparsity and regularization, and extensions to crossed effects, spatial data, and federated learning settings.


\backmatter

\bmhead{Supplementary Information}
This paper is accompanied with a supplementary text that includes omitted proofs and additional details about the reported simulated experiments. 

\bmhead{Acknowledgements}
We would like to acknowledge the late Sayan Mukherjee, whose thoughtfulness, generosity, and intellectual vision were instrumental in bringing together the authors of this work. Sayan had a remarkable ability to unite collaborators from very different backgrounds to tackle challenging and important problems, and this project is a direct reflection of that gift. His insight and encouragement shaped the direction of this work, and we are deeply grateful for his contributions and for the opportunity to have worked with him.
We also thank the anonymous referee for their thoughtful comments, which greatly enhanced the presentation of Sections 4 and 5.

\bmhead{Author Support}
Research reported in this publication was supported by the National Eye Institute of the National Institutes of Health (Bethesda, Maryland) under Awards Number R00EY033027 (SIB). The sponsor or funding organization had no role in the design or conduct of this research. The content is solely the responsibility of the authors and does not necessarily represent the official views of the National Institutes of Health. AA is member of INdAM (GNAMPA group), and acknowledges partial support of Dipartimento di Eccellenza, UNIPI, the Future of Artificial Intelligence Research (FAIR) foundation (WP2), PRIN project ConStRAINeD, PRA Project APRISE, and GNAMPA Project CUP\_E53C22001930001.

\bibliography{References/references}

\end{document}